\documentclass[journal,twoside]{IEEEtran}
\usepackage{graphicx}
\usepackage{balance}
\usepackage{cite}
\usepackage{amssymb,amsmath}
\usepackage{subfigure}
\usepackage[ruled,vlined]{algorithm2e}
\usepackage{amssymb,amsmath}
\usepackage{subfigure}
\usepackage[T1]{fontenc}
\usepackage[latin9]{inputenc}
\usepackage{float}
\usepackage[english]{babel}
\usepackage{epstopdf}
\usepackage{etoolbox}
\usepackage{algorithmic}
\usepackage{color}

\preto\subequations{\ifhmode\unskip\fi}
\makeatletter
\newcommand\restartchapters{\par
  \setcounter{chapter}{0}%
  \setcounter{section}{0}%
  \gdef\@chapapp{\chaptername}%
  \gdef\thechapter{\@arabic\c@chapter}}
\makeatother

\newtheorem{remark}{\bf Remark}

\newtheorem{lemma}{\bf Lemma}

\newtheorem{proposition}{\bf Proposition}

\newcommand{\st}{\mathrm{s.t.} }

\DeclareMathOperator{\tr}{tr}
\DeclareMathOperator{\rank}{rank}

\usepackage{etoolbox}

\begin{document}
\bstctlcite{IEEEexample:BSTcontrol}

\title{Enhancing PHY Security of Cooperative Cognitive Radio Multicast Communications}
\author{
\IEEEauthorblockN{Van-Dinh Nguyen, \textit{Student Member, IEEE,} Trung Q. Duong, \textit{Senior Member, IEEE,} Oh-Soon Shin, \textit{Member, IEEE,}	 Arumugam Nallanathan, \textit{Fellow, IEEE,} and 	George K. Karagiannidis, \textit{Fellow, IEEE}   } \\
 \thanks{V.-D.~Nguyen and O.-S.~Shin are with the School of Electronic Engineering and the Department of  ICMC Convergence Technology, Soongsil University, Seoul 06978, Korea (e-mail: \{nguyenvandinh, osshin\}@ssu.ac.kr).}
\thanks{T.~Q.~Duong is with the School of Electronics, Electrical Engineering and Computer Science, Queen's University Belfast, Belfast BT7 1NN, United Kingdom (e-mail: trung.q.duong@qub.ac.uk).}
\thanks{A.~Nallanathan is  with the Centre for Telecommunications Research, King's College London, London WC2R 2LS, U.K. (e-mail: arumugam.nallanathan@kcl.ac.uk).
}
\thanks{G.~K.~Karagiannidis is with the Department of Electrical and Computer Engineering, Aristotle University of Thessaloniki, Thessaloniki 54 124, Greece (e-mail: geokarag@auth.gr).
}
\thanks{Part of this work was presented at the  2017 IEEE International Conference on Communications (ICC) \cite{NguyenICC}.}
}
\maketitle
\vspace{-1.8cm}
\begin{abstract}
  In this paper, we propose a cooperative approach to improve the security of both primary and secondary systems in cognitive radio multicast communications. During their access to the frequency spectrum licensed to the primary  users, the secondary unlicensed users  assist the primary system in fortifying security by sending a jamming noise to the eavesdroppers, while  simultaneously protect themselves from eavesdropping.  The main objective of this work is to maximize the secrecy rate of the secondary system, while  adhering to all individual primary users' secrecy rate constraints.	In the case of active eavesdroppers and perfect channel state information (CSI)  at the transceivers,  the utility function of interest is nonconcave and the involved constraints are nonconvex, and thus, the optimal solutions are troublesome.  To solve this problem, we propose an iterative algorithm to arrive  at least to a local optimum of the original nonconvex problem. This algorithm is guaranteed to achieve a Karush-Kuhn-Tucker solution. Then, we extend the optimization approach to  the case of passive eavesdroppers and imperfect CSI knowledge at the transceivers, where the constraints  are transformed  into a linear matrix inequality and convex constraints, in order to facilitate the optimal solution.	
\end{abstract}
\begin{IEEEkeywords}
 Cognitive radio,   convex optimization, interference, jamming noise, secrecy capacity, multicast transmission.  
\end{IEEEkeywords}

\section{Introduction} \label{Introduction}

Traditionally, a secrecy mechanism is applied at the higher layers of a communication system by using a secret key exchange between the source and the destination, such as the Diffie-Hellman key exchange \cite{Diffie}. However, the execution of  key exchange algorithms  over wireless networks may be vulnerable to eavesdropping attacks, due to the broadcasting nature of the wireless transmission media. As a result,  research in information theory for wireless communication systems has focused on achieving secrecy, by using channel coding and signal processing techniques at the physical  layer (PHY)  \cite{Wyner, Cheong}.  Specifically, the pioneering work \cite{Wyner} introduced   PHY security via  wiretap channels,  by providing  perfect secrecy that can be attained when the eavesdropper channel is a degraded version of the main source-to-destination channel. 

Recently, PHY security for wireless communications has become  an important research area. The underlying idea is to guarantee a positive secrecy rate of legitimate users by exploiting the random characteristics of  wireless channel. In particular,  the authors in  \cite{Gopala} proposed a low-complexity on/off power allocation strategy to attain secrecy  under the assumption of full channel state information (CSI). The use of cooperative jamming noise (JN) was proposed in  \cite{Tekin}, where users who are prevented from transmitting according to a certain policy  block the eavesdropper and thereby assist the remaining users. {\color{black}In \cite{AnandTIFS10}, the authors analyzed the optimal location of an eavesdropper which results in zero secrecy capacity of all links, where the location is defined logically in terms of channel gains}. From a  quality-of-service (QoS) perspective, a secret transmit beamforming approach  was considered in  \cite{Mukherjee_1,Mukherjee_2,Liao}, in order to predetermine the signal-to-interference-plus-noise-ratio (SINR) target at the destination and/or at the eavesdropper. More recently, a jamming noise technique  (a.k.a. artificial noise)  was introduced, in order to improve the  secrecy capacity by confusing the decoding capability of the eavesdroppers \cite{Li,Lin_13,Gerbracht,Zhang,Romero,Yan,Yang_14}.  Furthermore, in  \cite{Goel},   a new secure transmission  was proposed in order to sustain the  secrecy of the communication, by utilizing   the available power to produce artificial noise for the eavesdropper.  The authors in \cite{Zhou} considered the case of a passive eavesdropper  with multi-antenna transmission, where the transmitter simultaneously transmits an information-bearing signal to the intended receiver and  artificial noise to the eavesdropper.  A  joint information and jamming beamforming technique for a full-duplex base station (BS)  which secures both  uplink and downlink transmission, was proposed in \cite{Zhu}. Finally, cooperation between the source and destination was proposed in \cite{Zheng}, with the destination operating in the full-duplex mode, i.e., the destination receives information from the source and sends a jamming signal to the eavesdropper at the same time.

Being a critical issue, PHY security of cognitive radio networks (CRNs), which deal with specific security risks due to the broadcasting nature of  radio signals,   has not been well investigated until  recently, e.g., in  \cite{Nguyen:TIFS:16,Pei,Yiyang,Gabry,He,Nguyen,Zhu:VT:15,Nguyen_15}. More specifically, in \cite{Nguyen:TIFS:16,Pei,Yiyang}, multi-antennas at the secondary transmitter  were utilized to attain  beamforming that maximizes the secrecy capacity of the secondary system,  while adhering to the peak interference constraint at the primary receiver. In \cite{Gabry}, cooperation  between the secondary system and the primary system was proposed, in order to improve the secrecy capacity of the primary system. Furthermore, a simple case with  single antenna at the eavesdropper was considered in \cite{He, Nguyen}. Particularly, in \cite{He},  joint beamforming for information and jamming noise was proposed to protect both  primary  and  secondary systems, with the secondary user acting as an amplify-and-forward relay to enhance the security of the primary system.   A jamming beamforming technique was designed in \cite{Nguyen}, based on the nullspace of the legitimate channel, in order to protect the primary system by treating the signal from the secondary transmitter as interference. In \cite{Zhu:VT:15}, the authors considered a CRN model, where both  primary user (PU) and  secondary user (SU) send their confidential messages to intended receivers that are surrounded by a single eavesdropper. {\color{black} Besides,  the capacity-equivocation
region of cognitive interference channel was  obtained in \cite{LiangIT09}, where the primary receiver is treated as untrusted user (eavesdropper) who intends to eavesdrop the  confidential message  of the secondary system. Extensions of \cite{LiangIT09} were made in  \cite{BafghiISITTA10,FarsaniISIT14} by additionally considering the secrecy of the primary system}.

In this paper, we consider the PHY security in cooperative cognitive radio multicast communications, where the eavesdroppers intend to wiretap  data from both  primary  and  secondary systems. We assume that the primary transmitter  is equipped only with a single antenna, which implies that the primary transmitter cannot generate a jamming signal or design a beamforming vector to protect itself from the eavesdroppers. The secrecy capacity of the primary system is improved by implementing a cooperative framework between the primary  and secondary systems. Specifically, the primary system allows the secondary system to share its spectrum, and in return the secondary system sends jamming noise to degrade the eavesdropper's channel, in order to protect the primary system. In the CRN multicast transmission model, we assume that there are one  group of PUs and $G$  groups of SUs, where all users in each group receives  identical information from its transmitter, and furthermore,  each group can be surrounded by multiple eavesdroppers. We note that the recent work in \cite{Zhu:VT:15} is a special case of this paper, where only a single receiver and a single eavesdropper are assumed, which is well-known as unicast mode.  

The aim of this paper is  to design the optimal beamforming vectors that realize the PHY security and maximize the secrecy rate of the secondary system, while ensuring adherence to the  individual secrecy rate constraints at each primary user.
 Specifically, the main contributions of this paper can be summarized as follows:
\begin{itemize}
	\item For the perfect CSI case, we design a joint information and jamming signal at the secondary transmitter, where information is intended for secondary receivers  and jamming noise is intended for eavesdroppers. The main objective is to maximize the secrecy rate of the secondary system, while satisfying the minimum secrecy rate requirement for each legitimate user of the primary system as well as  the  power constraint. We show that the equivalent problem can be converted to a single-layer optimization problem, which can be easily solved through convex quadratic programming.
	\item When the CSI of the channel from the secondary transmitter to the PUs is imperfect and only  partial CSI of the eavesdroppers is available, we  transform the non-linear constraints  into a linear matrix inequality and convex constraints, based on a specific matrix inequality lemma. We show that the approximate optimization problem  can be efficiently solved in a similar manner as the perfect CSI case.
	\item We propose an efficient method to find the approximate solution for optimal transmit beamforming,  by providing the convexity of the original problem that is considered through the use of   a convex approximation.  The optimal solutions of transmit beamforming for the confidential information and jamming noise do not fix the transmit strategy. Importantly, we  develop an iterative algorithm of low complexity for the computational solution of the considered optimization problem. The obtained solutions are proved to be at least local optimum, as satisfying the necessary optimal conditions. 
	\item We provide extensive numerical results to justify the novelty of the proposed algorithm and compare its performance with the known solutions. In particular, the numerical results demonstrate fast convergence of the proposed algorithm and a significant improvement of the secrecy rate, compared with other known solutions. We should remark that our results are more general than in \cite{Zhu:VT:15}, which was considered under the assumptions of one eavesdropper and perfect CSI. In addition, the model in this paper is of practical interest in designing networks that are required to transmit the same data to a group of users, for example, in video broadcasting and various applications. Moreover, the  considered problem in this paper is highly nonlinear and nonconvex function, thus it is more challenging to solve compared to SINR-based design in \cite{Zhu:VT:15}.
	\end{itemize}
	
The rest of this paper is organized as follows. Section \ref{System Model} describes the CRN multicast transmission model with multiple eavesdroppers and formulates the optimization problem. Section \ref{Perfect-CSI} derives optimal beamforming for information signal and jamming noise at the secondary transmitter   under the assumption of perfect CSI, while  Section \ref{Imperfect-CSI}  extends the considered problem to the case of imperfect CSI and  passive eavesdropper.  Section \ref{Numerical} provides  numerical results  and discussions. Finally, the conclusions are drawn in Section \ref{Conclusion}. {\color{black}In order to make the rest of the paper easy to follow, the notations and
symbols used in the paper are specified in Table \ref{TableNotations}}.

\section{System Model and Optimization Problem} \label{System Model}

\begin{figure}[t]
\centering
\includegraphics[trim=2.0cm 0cm 0.0cm 0.0cm, width=0.4\textwidth]{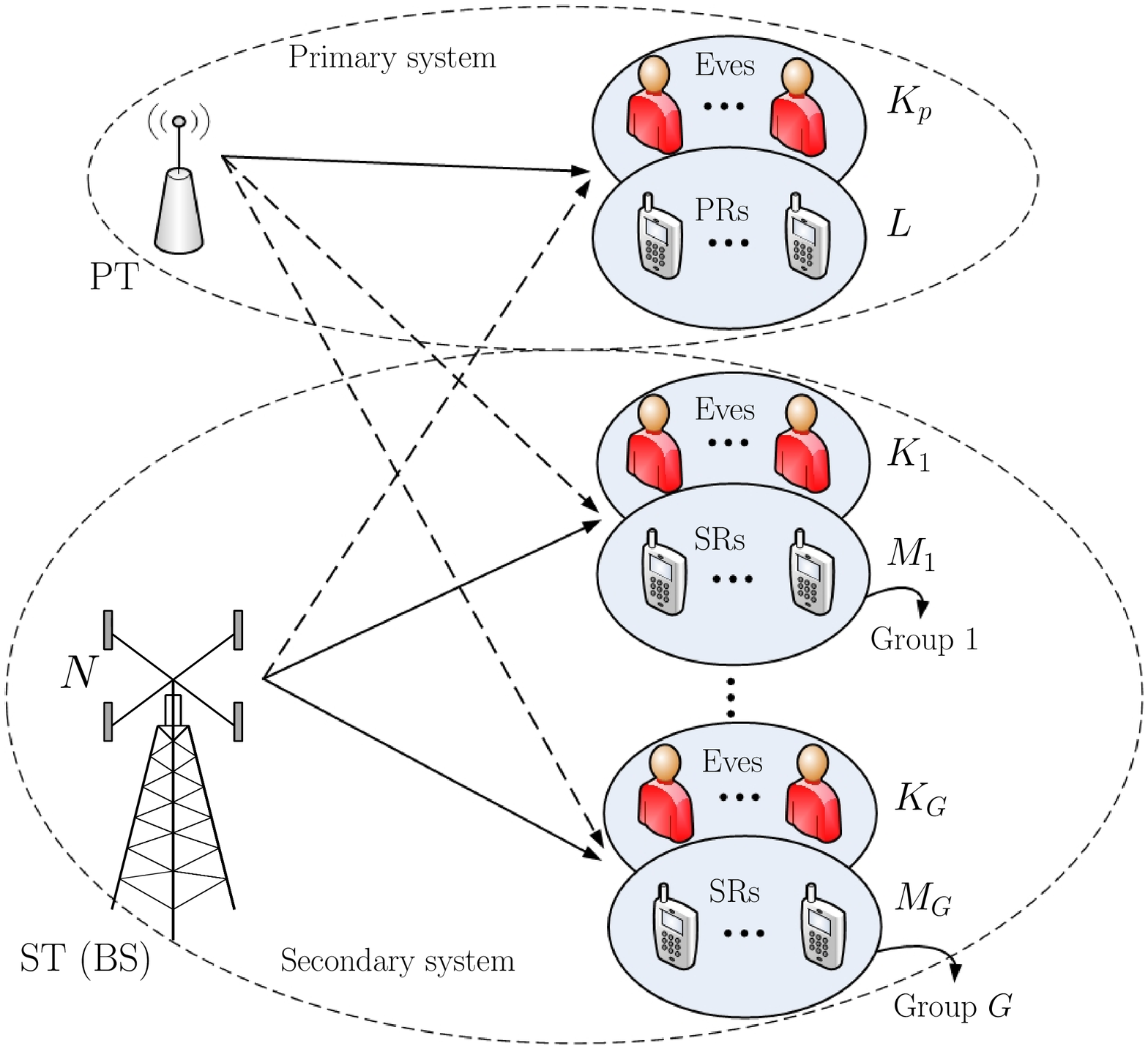}
\caption{A cooperative CRN multicast transmission  model with multiple eavesdroppers.}
\label{fig:zoneselect}
\end{figure}
\subsection{System  Model} \label{Cognitive Radio Network}
We consider the PHY security of CRN multicast transmission  with  cooperation  between a primary system and a secondary system. The primary system consists of one primary transmitter (PT) and $L$ primary receivers (PRs), while the secondary system consists of one secondary transmitter  (ST) and $M$ secondary receivers (SRs),  as illustrated in Fig. 1.   The ST, which is a BS,  is equipped with $N$ antennas,  whereas all  other nodes are equipped with only one antenna.\footnote{We note that the solution for multiple antennas at the PT is straightforward by following the same procedure presented in this paper since the resource allocation strategies at the ST and PT are independent.}  The opportunistic spectrum access is improved by assigning the ST  to send  $G$ information bearing signals $s_g, g=1, \cdots, G$, where $s_g$ is the information being sent to the $g$-th group  with unit average power $\mathbb{E}\{\left|s_g\right|^2\}=1$.  We assume that each individual multicast group $\mathcal{G}_g$ in the secondary system consists of $M_g$  secondary receivers. Specifically, the number of  SRs in  group $\mathcal{G}_g$ is denoted by $\mathcal{S}_g=\left\{1,\cdots, m_g,\cdots, M_g\right\}$. Then, the total number of SRs  in the secondary system with multicast transmission is indeed $M=\sum_{g=1}^GM_g$. In the multicast transmission, all  users within the same group will receive  identical data from its transmitter.   
Regarding security, we assume that the eavesdroppers (Eves) potentially intend to wiretap and decode confidential messages from both  primary  and  secondary systems \cite{Mokari}.  We assume that each group $\mathcal{G}_g$ and the PRs are respectively wiretapped by a set of Eves such as $\mathcal{K}_{e,g}\triangleq\left\{1,\cdots,k_g,\cdots,K_g\right\}$, $\forall g$ and $\mathcal{K}_{p}\triangleq\left\{1,\cdots,k_p,\cdots,K_p\right\}$. This implies that at the same time, each legitimate user is wiretapped by a separate  group of Eves.

\begin{table}[t]
\centering
\caption{Notations and Symbols}
\label{TableNotations}
{\color{black}\begin{tabular}{ll}\hline
$\mathbf{X}^{H}$, $\mathbf{X}^{T}$ and $\tr(\mathbf{X})$      & Hermitian transpose, normal transpose \\
                                                               & and trace of a matrix $\mathbf{X}$                                                             \\
$\|\cdot\|$ and $|\cdot|$                                      & Euclidean norm of a matrix or vector \\
                                                               &and the magnitude of a complex scalar                                                            \\
$\mathbf{I}_N$                                                 & $N\times N$ identity matrix                                                                                                           \\
$\mathbf{x}\sim\mathcal{CN}(\boldsymbol{\eta},\boldsymbol{Z})$ &Random vector following a complex circular \\
                                                               &Gaussian distribution with mean $\boldsymbol{\eta}$ \\
																															 &and covariance matrix $\boldsymbol{Z}$ \\
$\mathbb{E}[\cdot]$                                            &Statistical expectation                                                                                                               \\
$\mathbf{X}\succeq\mathbf{0}$                                  &Positive semidefinite matrix                                                                                                          \\
$\Re\{\cdot\}$                                                 &Real part of the argument                                                                                                             \\
$\mathbf{h}_{m_g}$ and $\mathbf{f}_{l}$                        &Channels from ST to $m_g$-th $\mbox{SR}$ and $l$-th $\mbox{PR}$                                                                       \\
$\mathbf{g}_{k_g}$ and $\mathbf{f}_{k_p}$                &Channels from ST to $k_g$-th $\mbox{Eve}$ and $k_p$-th $\mbox{Eve}$                                                                   \\
$h_{l}$ and $f_{m_g}$                                          &Channels from PT to $l$-th $\mbox{PR}$ and  $m_g$-th $\mbox{SR}$                                                         \\
$g_{k_p}$ and $f_{k_g}$                                        &Channels from PT to $k_p$-th $\mbox{Eve}$ and $k_g$-th $\mbox{Eve}$                                                                   \\
$\mathbf{w}_g$                                                 &Beamforming vector at ST intended to group $\mathcal{G}_g$                                                                            \\
$\mathbf{u}$                                                   &Artificial noise vector with $\mathbf{u}\sim\mathcal{CN}(\mathbf{0}, \mathbf{U}\mathbf{U}^H)$                                        \\
$t_g$                                                           &Maximum allowable rate for $k_g$-th $\mbox{Eve}$                                                                             \\
$z$                                                & Maximum allowable rate  for $k_p$-th Eve                                                                                      \\
$\varphi$                                                      &Objective variable in maximizing secrecy rate \\
                                                               &of secondary system                                                                        \\
$\alpha$                                                       &Minimum SINR requirement for $l$-th PR \\ 
$\phi_g$                                                       &Maximum received SINR for  $k_g$-th $\mbox{Eve}$\\
$\beta$                                                       &Maximum received SINR for $k_p$-th Eve \\                                                                                                                    
\hline		\end{tabular}}
\end{table}

We aim to design multiple beamforming vectors at  the ST, one for the JN and the other  for its own information signal, to protect both  primary and secondary systems.   The transmit power at the PT is $P_p >0$ and the data intended for the PRs is $x_p$ with unit average power $\mathbb{E}\{\left|x_p\right|^2\}=1$. Before transmission, the data of the SRs $s_g$ in the group $\mathcal{G}_g$ is weighted to the $N\times 1$ beamforming vector $\mathbf{w}_g$, $\forall g$.  Hence,  the transmitted signals at the ST can be expressed through a vector $\mathbf{x}_s$ as  
\begin{equation}
\mathbf{x}_s = \sum_{g=1}^G\mathbf{w}_gs_g+\mathbf{u}
\label{eq:xs}
\end{equation}
where  $\mathbf{u}$ is the artificial noise vector, whose elements are zero-mean complex Gaussian random variables with covariance matrix $\mathbf{U}\mathbf{U}^{H}$, such that $\mathbf{u}\sim	\mathcal{CN}(\mathbf{0}, \mathbf{U}\mathbf{U}^H)$ with $\mathbf{U}\in\mathbb{C}^{N\times N}$. The artificial noise  $\mathbf{u}$ is assumed to be unknown to all SRs, PRs, and Eves. For notational simplicity, we define $\mathbf{w}\triangleq[\mathbf{w}_{1}^T, \mathbf{w}_{2}^T,\cdots,\mathbf{w}_{G}^T]^T\in\mathbb{C}^{NG\times 1}$.

The corresponding SINR at the $l$-th $\mbox{PR}$ for $l=1,\cdots, L$ and the $k_p$-th $\mbox{Eve}$  for $k_p=1,\cdots, K_p$  are respectively given by
\begin{IEEEeqnarray}{rCl}
\Gamma_{p,l}(\mathbf{w},\mathbf{U})&=&\frac{P_p|h_l|^2}{\sum_{g=1}^G|\mathbf{f}_l^{H}\mathbf{w}_g|^2+ \|\mathbf{f}_l^{H}\mathbf{U}\|^2+\sigma_l^2},\label{eq:SINR:pr}\\
\Gamma_{e,k_p}(\mathbf{w},\mathbf{U})&=&\frac{P_p|g_{k_p}|^2}{\sum_{g=1}^G|\mathbf{f}_{k_p}^{H}\mathbf{w}_g|^2+ \|\mathbf{f}_{k_p}^{H}\mathbf{U}\|^2+\sigma_{k_p}^2}\label{eq:SINR:pre}
\end{IEEEeqnarray} 
where $h_{l}\in\mathbb{C}$, $g_{k_p} \in\mathbb{C}$, $\mathbf{f}_{l} \in\mathbb{C}^{N\times 1}$, and $\mathbf{f}_{k_p}\in\mathbb{C}^{N\times 1}$ are  the respective baseband equivalent channels of the links PT $\rightarrow$ $l$-th $\mbox{PR}$, PT $\rightarrow$ $k_p$-th $\mbox{Eve}$, ST $\rightarrow$ $l$-th $\mbox{PR}$, and ST $\rightarrow$ $k_p$-th $\mbox{Eve}$. $\sigma_l^2$ and $\sigma_{k_p}^2$ are the variance of the additive white Gaussian noise (AWGN) at the $l$-th $\mbox{PR}$  and  ${k_p}$-th $\mbox{Eve}$, respectively.

The respective SINR at the $m_g$-th $\mbox{SR}$  in the group $\mathcal{G}_g$ and the $k_g$-th  $\mbox{Eve}$  are given by
\begin{IEEEeqnarray}{rCl}\label{eq:SINR_ekg}
&&\Gamma_{s,m_g}(\mathbf{w},\mathbf{U})= \nonumber\\
&&\quad\frac{|\mathbf{h}_{m_g}^{H}\mathbf{w}_g|^2}{\sum_{i=1, i\neq g }^G|\mathbf{h}_{m_g}^{H}\mathbf{w}_i|^2+ \|\mathbf{h}_{m_g}^{H}\mathbf{U}\|^2+P_p|f_{m_g}|^2+\sigma_{m_g}^2},\label{eq:SINR:sr}\qquad\\
&&\Gamma_{e,k_g}(\mathbf{w},\mathbf{U})= \nonumber\\
&&\qquad\frac{|\mathbf{g}_{k_g}^{H}\mathbf{w}_g|^2}{\sum_{i=1,i\neq g}^G|\mathbf{g}_{k_g}^{H}\mathbf{w}_i|^2+\|\mathbf{g}_{k_g}^{H}\mathbf{U}\|^2+P_p|f_{k_g}|^2+\sigma_{k_g}^2}\label{eq:SINR:se}
\end{IEEEeqnarray} 
where $\mathbf{h}_{m_g}\in\mathbb{C}^{N\times 1}$, $\mathbf{g}_{k_g}\in\mathbb{C}^{N\times 1}$, $f_{m_g}\in\mathbb{C}$, and $f_{k_g} \in\mathbb{C}$ are the corresponding baseband equivalent channels of the links ST $\rightarrow$ $m_g$-th $\mbox{SR}$, ST $\rightarrow$ $k_g$-th $\mbox{Eve}$, PT $\rightarrow$ $m_g$-th $\mbox{SR}$, PT $\rightarrow$ $k_g$-th $\mbox{Eve}$. $\sigma_{m_g}^2$ and $\sigma_{k_g}^2$ are the variance of AWGN at the $m_g$-th $\mbox{PR}$  and $k_g$-th $\mbox{Eve}$, respectively. We further assume that all channels remain constant during a transmission block, yet change independently from one block
to another. {\color{black} By using dirty-paper coding (DPC), the ST with encoding order from the group $\mathcal{G}_1$ to $\mathcal{G}_G$ enables the SRs in  $\mathcal{S}_g$ to know the information signals intended for the SRs in $\mathcal{S}_{g'}, g'= 1,\cdots,g-1$  non-casually,  so that it can
be perfectly eliminated \cite{Gamal}. Hence, the SINR in \eqref{eq:SINR:sr} by DPC can be rewritten as
\begin{IEEEeqnarray}{rCl}\label{eq:SINR_ekgDPC}
&&\Gamma_{s,m_g}^{\mathtt{DPC}}(\mathbf{w},\mathbf{U})= \nonumber\\
 &&\qquad\qquad \frac{|\mathbf{h}_{m_g}^{H}\mathbf{w}_g|^2}{\sum_{ i> g }|\mathbf{h}_{m_g}^{H}\mathbf{w}_i|^2+ \|\mathbf{h}_{m_g}^{H}\mathbf{U}\|^2+P_p|f_{m_g}|^2+\sigma_{m_g}^2}.\nonumber
\end{IEEEeqnarray} 
 It is clear that under the same beamformer/precoder $(\mathbf{w},\mathbf{U})$, $\Gamma_{s,m_g}^{\mathtt{DPC}}(\mathbf{w},\mathbf{U})$ is better than $\Gamma_{s,m_g}(\mathbf{w},\mathbf{U})$. However, DPC is difficult to implement in practice  due to its extremely high computational complexity and thus  remains only
as a theoretical bound.}

{\color{black}The channel of each legitimate user together with the respective Eves form a compound wiretap channel \cite{Lian:EUR:09}}. Therefore, the achievable secrecy rate for the $l$-th PR  of the primary system, denoted by $C_{p,l}(\mathbf{w},\mathbf{U})$, can be expressed as \cite{Lian:EUR:09,CsiszarIT78}
\begin{IEEEeqnarray}{rCl}\label{eq:C_pr}
C_{p,l}(\mathbf{w},\mathbf{U})=&&\Bigl[\log_2\bigl(1+\Gamma_{p,l}(\mathbf{w},\mathbf{U})\bigr)\nonumber\\
&&\qquad -\;\underset{k_p\in\mathcal{K}_p}{\max}\log_2\bigl(1+\Gamma_{e,k_p}(\mathbf{w},\mathbf{U})\bigr)\Bigr]^+
\end{IEEEeqnarray} 
where  $\left[x\right]^+=\max\left\{0,x\right\}$.

Similarly, the achievable secrecy rate for the  $m_g$-th SR of the secondary system, denoted by $C_{s,m_g}(\mathbf{w},\mathbf{U})$, can be expressed as \cite{Gopala}
\begin{IEEEeqnarray}{rCl}\label{eq:C_sk}
C_{s,m_g}(\mathbf{w},\mathbf{U})=&&\Bigl[\log_2\bigr(1+\Gamma_{s,m_g}(\mathbf{w},\mathbf{U})\bigr)\nonumber\\
&&\quad -\;  \underset{k_g\in\mathcal{K}_{e,g}}{\max}\log_2\bigr(1+\Gamma_{e,k_g}(\mathbf{w},\mathbf{U})\bigr)\Bigr]^+.
\end{IEEEeqnarray}

If $C_{p,l}(\mathbf{w},\mathbf{U})$ and $C_{s,m_g}(\mathbf{w},\mathbf{U})$ are above zero,  the signal transmitted from the PT and ST are determined to be ``undecodable'' as is indicated in \cite{Tekin}. 

\subsection{Optimization Problem Formulation}
The objective of the system design is to maximize the minimum (max-min) secrecy rate  of the secondary system  while satisfying  the minimum QoS requirements, such as the  secrecy rate achievable for the primary system. Accordingly, the optimization problem can be mathematically formulated as
\begin{IEEEeqnarray}{rCl}\label{eq:problem_1}
\mathbf{P.1}:\quad\underset{\mathbf{w}, \mathbf{U}}{\mathrm{\max}}&&\mathop{\mathrm{\min}}\limits_{m_g\in\mathcal{S}_{g},   g\in\mathcal{G}}
                   \quad C_{s,m_g}(\mathbf{w},\mathbf{U})\IEEEyessubnumber\label{eq:8a}\\
  \st&&\quad C_{p,l}(\mathbf{w},\mathbf{U})\geq \bar{R}_{p,l},\, l\in\mathcal{L} \IEEEyessubnumber\label{eq:8b}\\
     		  && \quad  \sum\nolimits_{g=1}^G\|\mathbf{w}_g\|^2+ \|\mathbf{U}\|^2 \leq P_{s}\, \IEEEyessubnumber\label{eq:8c}
\end{IEEEeqnarray}
where $\mathcal{L}\triangleq\{1,\cdots, L\}$ and $\mathcal{G}\triangleq\{1,\cdots, G\}$. In \eqref{eq:8b}, $\bar{R}_{p,l} > 0$ are the minimum secrecy rate requirement for each legitimate user  of the primary system. This implies that the QoS for each PR can be different and flexible. In \eqref{eq:8c}, $P_s$ is the transmit power budget at the ST.

\begin{remark}
There are two other performance metrics of interest involved in the considered system. In particular,  one is to maximize the secrecy rate of the primary system subject to the secrecy rate threshold of secondary system and the transmit power budget at the ST, while the other is to minimize the total transmit power at the ST subject to the secrecy rate threshold of both systems. However, the optimal solution for \eqref{eq:problem_1} is also applicable to those cases that will be presented shortly.
\end{remark}

The recent works in \cite{Zhu:VT:15, Zhu, Lui, Ng} often introduce  new variables to relax the optimization problem as
\begin{equation}
\widetilde{\mathbf{W}}_g = \mathbf{w}_g\mathbf{w}_g^H, \forall g
\end{equation}
which must satisfy the rank-one constraint, i.e., $\rank(\widetilde{\mathbf{W}}_g) =1, \forall g$. Then, they use   semi-definite  program (SDP) relaxation to solve the optimization problem by constructing an equivalent problem. In which, the optimal solution involves the dual variables of the relaxed problem. Unfortunately, some numerical solvers may not exhibit the optimal solution of dual variables, and then the construction of primal variables may not be possible. In what follows, we will solve \eqref{eq:problem_1} via a convex quadratic program and thus the rank-one constraints are automatically satisfied.

\section{Theoretical Benchmark  With Perfect CSI}\label{Perfect-CSI}
We first consider the case for which the instantaneous CSI of all  channels is available at the transceivers. In particular, the CSI of all channels in both  systems can be obtained through feedback from the legitimate receivers to the legitimate transmitters.  {\color{black}After CSI acquisition, we assume that only $M$ SRs and $L$ PRs are scheduled to be concurrently served. Herein, the remaining users (unscheduled users) are not necessarily malicious, but they could be
 untrusted users. Thus, the unscheduled users are treated as potential eavesdroppers, but with perfectly known CSI at the transmitters.} These assumptions are consistent with several  previous works  on information theoretic analysis and optimization for the similar kind of problem, \cite{Gopala,Tekin,Lin_13,Zhu, Zheng}, for instance.\footnote{Though  this assumption is quite ideal, however, the performance with assumption of perfect CSI is still of practical importance since it plays as a benchmark how the CRN system may achieve in more realistic conditions \cite{Zhu:VT:15, Nguyen, Nguyen_15, Yiyang}. } 
\subsection{Optimal Solution}\label{sc:construcsolution}

We note that finding an optimal solution for \eqref{eq:problem_1} is challenging due to the nonconcavity of the objective function and nonconvexity of the feasible set. In this section, we propose an iterative algorithm that arrives a local optimum of the considered optimization problem. As the first step, we convert \eqref{eq:problem_1} to another equivalent form as
\begin{IEEEeqnarray}{lCl}\label{eq:rew:1}
&&\underset{\mathbf{w}, \mathbf{U}, \boldsymbol{t}, z}{\mathrm{\mathrm{maximize}}}\mathop{\mathrm{\min}}\limits_{m_g\in\mathcal{S}_{g},g\in\mathcal{G}}                    \left\{\log_2\bigl(1+\Gamma_{s,m_g}(\mathbf{w},\mathbf{U})\bigl) - t_{g}\right\} \IEEEyessubnumber\label{eq:rew:a}\\
  &&\st\ \log_2\bigl(1+\Gamma_{e,k_g}(\mathbf{w},\mathbf{U})\bigl) \leq t_g,\ k_g\in\mathcal{K}_{e,g}, g\in\mathcal{G} \IEEEyessubnumber\label{eq:rew:b} \quad\\         
	&&\qquad \log_2\bigl(1+\Gamma_{p,l}(\mathbf{w},\mathbf{U})\bigr) - z\geq \bar{R}_{p,l},\ l\in\mathcal{L} \IEEEyessubnumber\label{eq:rew:c}\\
	&&  \qquad \log_2\bigl(1+\Gamma_{e,k_p}(\mathbf{w},\mathbf{U})\bigl) \leq z,\ k_p\in\mathcal{K}_p\IEEEyessubnumber\label{eq:rew:d}\\
     		  &&\qquad  \eqref{eq:8c}\IEEEyessubnumber\label{eq:rew:e}
\end{IEEEeqnarray}
where $\boldsymbol{t}\triangleq \{t_{g}\}$ and $z$ are the maximum allowable rates for Eves to wiretap the information signals from the ST and the PT, respectively.
The equivalence of \eqref{eq:problem_1} and \eqref{eq:rew:1} can be easily confirmed by justifying that  the constraint  \eqref{eq:rew:b} must hold with equality at optimum. We now provide a sketch of the proof to verify this
point.  Suppose that $\log_2\left(1+\Gamma_{e,k_g}(\mathbf{w},\mathbf{U})\right) < t_g$  for some $k_g$, there exist  the positive constants, i.e., $\Delta t_g > 0$  such that $\log_2\left(1+\Gamma_{e,k_g}(\mathbf{w},\mathbf{U})\right) = t_g - \Delta t_g$. As a result, $t_g - \Delta t_g$   is feasible to \eqref{eq:rew:1} but yielding a strictly larger objective. Thus, this is a contradiction to the optimality assumption. Even after the above transformations, \eqref{eq:rew:1} is still nonconvex and difficult to solve due to nonconcavity of the objective function. Toward a tractable form, let us rewrite \eqref{eq:rew:1} equivalently as
\begin{IEEEeqnarray}{rCl}\label{eq:rew:2}
&&\underset{\mathbf{w}, \mathbf{U}, \boldsymbol{t}, z, \varphi}{\mathrm{\mathrm{maximize}}}\quad  \varphi\IEEEyessubnumber\label{eq:rew:2:a}\\
   \st &&\quad    \log_2\bigl(1+\Gamma_{s,m_g}(\mathbf{w},\mathbf{U})\bigl) - t_{g} \geq  \varphi, m_g\in\mathcal{S}_{g},g\in\mathcal{G} \IEEEyessubnumber\label{eq:rew:2:b} \qquad\\
	&& \quad\log_2\bigl(1+\Gamma_{e,k_g}(\mathbf{w},\mathbf{U})\bigl) \leq t_g,\ k_g\in\mathcal{K}_{e,g}, g\in\mathcal{G} \IEEEyessubnumber\label{eq:rew:2:c} \\         
	&& \quad\log_2\bigl(1+\Gamma_{p,l}(\mathbf{w},\mathbf{U})\bigl) - z\geq \bar{R}_{p,l},\ l\in\mathcal{L} \IEEEyessubnumber\label{eq:rew:2:d}\\
	&&  \quad \log_2\bigl(1+\Gamma_{e,k_p}(\mathbf{w},\mathbf{U})\bigl) \leq z,\ k_p\in\mathcal{K}_p\IEEEyessubnumber\label{eq:rew:2:e}\\
     		  &&\quad  \eqref{eq:8c} \IEEEyessubnumber\label{eq:rew:2:f}
\end{IEEEeqnarray}
where $\varphi$ is newly introduced variable to maximize the secrecy rate of the secondary system. Observe that the objective function is monotonic in its argument, therefore, we now only deal with the nonconvex constraints \eqref{eq:rew:2:b}-\eqref{eq:rew:2:e}. Toward this end, we provide the following result.\footnote{Hereafter, suppose the value of $(\mathbf{w},\mathbf{U})$ at the $(n+1)$-th iteration in an iterative algorithm presented shortly is denoted by $(\mathbf{w}^{(n)},\mathbf{U}^{(n)})$.} 
{\color{black}\begin{lemma}\label{LemmaPCSI}
For the secondary system, the inner convex approximations of nonconvex constraints \eqref{eq:rew:2:b} and \eqref{eq:rew:2:c} are given by:
\begin{IEEEeqnarray}{rCl}
\mathcal{F}_{m_g}^{(n)}(\mathbf{w},\mathbf{U}) &\geq& (\varphi + t_g)\ln2\label{eq:rew:2:b:equi},\\
\mathcal{F}_{k_g}^{(n)}(\mathbf{w},\mathbf{U}) &\leq&  t_g\ln2\label{eq:sesinr:3}
\end{IEEEeqnarray}
where $\mathcal{F}_{m_g}^{(n)}(\mathbf{w},\mathbf{U})$ and  $\mathcal{F}_{k_g}^{(n)}(\mathbf{w},\mathbf{U})$ are a lower bounding concave function for $\log_2\bigl(1+\Gamma_{s,m_g}(\mathbf{w},\mathbf{U})\bigl)$ and an upper bounding convex function for $\log_2\bigl(1+\Gamma_{e,k_g}(\mathbf{w},\mathbf{U})\bigl)$, which are concretized by \eqref{eq:srlog:2} and \eqref{eq:sesinr:4} in Appendix A, respectively.

Similarly for the primary system, the nonconvex constraints \eqref{eq:rew:2:d} and \eqref{eq:rew:2:e} are innerly approximated by
the following convex constraints:
\begin{IEEEeqnarray}{rCl}\label{eq:prsinr:1}
\mathcal{P}_{l}^{(n)}(\mathbf{w},\mathbf{U}) &\geq&  (z + \bar{R}_{p,l})\ln2\label{eq:prsinr:1a},\\
\mathcal{P}_{k_p}^{(n)}(\mathbf{w},\mathbf{U}) &\leq&  z\ln2\label{eq:prsinr:1b}
\end{IEEEeqnarray}
where $\mathcal{P}_{l}^{(n)}(\mathbf{w},\mathbf{U})$ and  $\mathcal{P}_{k_p}^{(n)}(\mathbf{w},\mathbf{U})$ are a lower bounding concave function for $\log_2\bigl(1+\Gamma_{p,l}(\mathbf{w},\mathbf{U})\bigl)$ and an upper bounding convex function for $\log_2\bigl(1+\Gamma_{e,k_p}(\mathbf{w},\mathbf{U})\bigl)$, which are also concretized by \eqref{eq:prsinr:2} and \eqref{eq:prsinr:3} in Appendix A, respectively.
\end{lemma}
\begin{proof}
See Appendix A.
\end{proof}
It is noteworthy that the following equalities hold at the optimum, i.e., $(\mathbf{w}^{(n+1)},\mathbf{U}^{(n+1)}) = (\mathbf{w}^{(n)},\mathbf{U}^{(n)})$:
\begin{IEEEeqnarray}{rCl}
\mathcal{F}_{m_g}^{(n)}(\mathbf{w}^{(n)},\mathbf{U}^{(n)})&=& \log_2\Bigl(1+\Gamma_{s,m_g}\bigl(\mathbf{w}^{(n)},\mathbf{U}^{(n)}\bigr)\Bigl),\label{eq:srlog:3}\\
\mathcal{F}_{k_g}^{(n)}(\mathbf{w}^{(n)},\mathbf{U}^{(n)}) &=& \log_2\Bigl(1+\Gamma_{e,k_g}\bigl(\mathbf{w}^{(n)},\mathbf{U}^{(n)}\bigl)\Bigl),\\
\mathcal{P}_{l}^{(n)}(\mathbf{w}^{(n)},\mathbf{U}^{(n)})   &=& \log_2\Bigl(1+\Gamma_{p,l}\bigl(\mathbf{w}^{(n)},\mathbf{U}^{(n)}\bigl)\Bigl),\\
\mathcal{P}_{k_p}^{(n)}(\mathbf{w}^{(n)},\mathbf{U}^{(n)}) &=& \log_2\Bigl(1+\Gamma_{e,k_p}\bigl(\mathbf{w}^{(n)},\mathbf{U}^{(n)}\bigl)\Bigl).  
\end{IEEEeqnarray}}

In summary, at the $(n+1)$-th iteration of the proposed method, we solve the following convex problem
\begin{IEEEeqnarray}{rCl}\label{eq:convexapp:1}
&&\underset{\mathbf{w}, \mathbf{U}, \boldsymbol{t}, z, \varphi}{\mathrm{\mathrm{maximize}}}\quad  \varphi\IEEEyessubnumber\label{eq:convexapp:1:a}\\
   \st &&\quad    \mathcal{F}_{m_g}^{(n)}(\mathbf{w},\mathbf{U}) \geq (\varphi + t_g)\ln2,\ m_g\in\mathcal{S}_{g},g\in\mathcal{G} \IEEEyessubnumber\label{eq:convexapp:1:b} \qquad\\
	&& \quad \mathcal{F}_{k_g}^{(n)}(\mathbf{w},\mathbf{U}) \leq  t_g\ln2,\ k_g\in\mathcal{K}_{e,g}, g\in\mathcal{G} \IEEEyessubnumber\label{eq:convexapp:1:c} \\ 	&& \quad\mathcal{P}_{l}^{(n)}(\mathbf{w},\mathbf{U}) \geq  (z + \bar{R}_{p,l})\ln2,\ l\in\mathcal{L} \IEEEyessubnumber\label{eq:convexapp:1:d}\\
	&&  \quad \mathcal{P}_{k_p}^{(n)}(\mathbf{w},\mathbf{U}) \leq  z\ln2,\ k_p\in\mathcal{K}_p  \IEEEyessubnumber\label{eq:convexapp:12:e}\\
     		  &&\quad  \eqref{eq:8c}. \IEEEyessubnumber\label{eq:convexapp:1:f}
\end{IEEEeqnarray}
{\color{black}An iterative algorithm for solving \eqref{eq:convexapp:1} requires an initial feasible point of \eqref{eq:rew:2} to start, i.e., the constraints \eqref{eq:rew:2:d}-\eqref{eq:rew:2:f} are satisfied}. Therefore, we solve the following nonconvex optimization problem
\begin{IEEEeqnarray}{rCl}\label{ini2.m}
\max_{\mathbf{w}, \mathbf{U},  z}\;&&\min_{l\in\mathcal{L}}\;
\Bigl\{\log_2\bigl(1+\Gamma_{p,l}(\mathbf{w},\mathbf{U})\bigl) - z -  \bar{R}_{p,l}
  \Bigr\}\IEEEyessubnumber \label{eq:ini2:a}\\
	&&\st\quad \log_2\bigl(1+\Gamma_{e,k_p}(\mathbf{w},\mathbf{U})\bigl) \leq z,\ k_p\in\mathcal{K}_p\IEEEyessubnumber \label{eq:ini2:b}\\
	  &&\qquad\ \,  \eqref{eq:8c}. \IEEEyessubnumber \label{eq:ini2:c}
\end{IEEEeqnarray}
We first generate a feasible point $(\mathbf{w}^{(0)}, \mathbf{U}^{(0)})$ to satisfy \eqref{eq:ini2:c} and then solve the following convex approximation problem at the $n$-th iteration
\begin{IEEEeqnarray}{rCl}\label{ini3.m}
\max_{\mathbf{w}, \mathbf{U},  z}\;&&\min_{l\in\mathcal{L}}\;
\Bigl\{\mathcal{P}_{l}^{(n)}(\mathbf{w},\mathbf{U}) -  (z + \bar{R}_{p,l})\ln2
  \Bigr\}\IEEEyessubnumber \label{eq:ini3:a}\\
	&&\st\quad \mathcal{P}_{k_p}^{(n)}(\mathbf{w},\mathbf{U}) \leq  z\ln2,\ k_p\in\mathcal{K}_p\IEEEyessubnumber \label{eq:ini3:b}\\
	  &&\qquad\ \, \eqref{eq:8c} \IEEEyessubnumber \label{eq:ini3:c}
\end{IEEEeqnarray}
and  output a feasible point of \eqref{eq:rew:2} when
\begin{equation}\label{ini4.m}
\min_{l\in\mathcal{L}}\;
\Bigl\{\mathcal{P}_{l}^{(n)}(\mathbf{w},\mathbf{U}) -  (z + \bar{R}_{p,l})\ln2
  \Bigr\}\geq 0.
\end{equation}
We numerically observe that it requires no more than 3 iterations to satisfy \eqref{ini4.m} in all cases.
After solving \eqref{eq:convexapp:1},  we update $(\mathbf{w}^{(n)}, \mathbf{U}^{(n)})$ for the next iteration until convergence or maximum
required number of iterations. Algorithm \ref{algo:proposed:DUAL} outlines the proposed iterative method for solving \eqref{eq:problem_1}.

\begin{algorithm}[t]
\begin{algorithmic}[1]

\protect\caption{An iterative algorithm to solve \eqref{eq:problem_1}}

\label{algo:proposed:DUAL}

\global\long\def\algorithmicrequire{\textbf{Initialization:}}

\REQUIRE  Set $n:=0$ and solve \eqref{ini3.m} to generate an  initial feasible point $\bigl(\mathbf{w}^{(n)},\mathbf{U}^{(n)}\bigr)$

\REPEAT
\STATE Solve \eqref{eq:convexapp:1} to obtain the optimal solution: $\bigl(\mathbf{w}^{*},\mathbf{U}^{*}\bigr)$.

\STATE Update\ $\mathbf{w}^{(n+1)}:=\mathbf{w}^{*}$  and  $\mathbf{U}^{(n+1)}:=\mathbf{U}^{*}$.

\STATE Set $n:=n+1.$
\UNTIL Convergence or maximum required number of iterations\\
\end{algorithmic} \end{algorithm}

\subsection{Proof of Convergence and Complexity Analysis}
The convergence result of Algorithm \ref{algo:proposed:DUAL} is stated in the following proposition.
\begin{proposition}\label{prop1} Algorithm~\ref{algo:proposed:DUAL} produces a sequence $\bigl\{ \bigl(\mathbf{w}^{(n)},\mathbf{U}^{(n)}\bigl)\bigr\}$
of improved points of \eqref{eq:problem_1}, which converges to a Karush-Kuhn-Tucker (KKT) point.
\end{proposition}

\begin{IEEEproof} See Appendix B.
\end{IEEEproof}

\textit{Complexity Analysis}:  We note that the proposed iterative algorithm requires solving only simple convex quadratic and linear constraints at each iteration.  We now provide the complexity analysis of Algorithm~\ref{algo:proposed:DUAL}. Specifically, in each iteration of Algorithm~\ref{algo:proposed:DUAL}, the per-iteration computational complexity of solving \eqref{eq:convexapp:1} is $\mathcal{O}(n^2\tilde{n}^{2.5}+\tilde{n}^{3.5})$, where $n = N(G+N) + G + 2$ is scalar real variables and $\tilde{n} = \sum_{g=1}^G(M_g+K_g) + K_p + L + 1$ is quadratic and linear constraints \cite{Ben:2001}.

\section{Optimal Solution With Realistic Scenario}\label{Imperfect-CSI}

\subsection{CSI Model}
In this section, we extend the optimization approach of the last section to a realistic scenario, where the instantaneous CSI between  ST and PRs is  imperfectly known and  Eves are  passive devices. Specifically, the primary and  secondary systems may not  cooperate completely in reality, and therefore the channels $\mathbf{f}_l, \forall l$ will be difficult to obtain perfectly. For instance, the PRs may be inactive  for a long period of the secondary data transmission time. Then, the CSI of the PRs can be only obtained   at the ST when the PRs is in active mode with the PT. As a result, the CSI of PRs at the ST
may be outdated when the secondary system performs the transmit strategy.  Hence, the CSI of the link between the ST and PRs is modeled as \cite{Li}
\begin{equation}\begin{aligned}
&\mathbf{f}_l=\mathbf{\hat{f}}_l+\Delta\mathbf{f}_l,\; \forall l\\
&\Omega_l\triangleq\{\Delta{\mathbf{f}}_l\in \mathbb{C}^{N\times 1}:\Delta{\mathbf{f}}_l^{H}\Delta{\mathbf{f}}_l\leq \delta^2_l\}
\end{aligned}\label{eq:imperfect:channel}\end{equation}
where $\mathbf{\hat{f}}_l$ is the channel estimate of the $l$-th PR  available at the ST,  and $\Delta\mathbf{f}_l$ represents the associated CSI error. In particular, we assume a time division duplex system with slowly time-varying channels. At the beginning of each time slot, the legitimate users (PRs, SRs)  report their channel gains to the ST. The downlink CSI of the ST-to-legitimate users are obtained by measuring the uplink pilot based on some estimation methods, such as minimum-mean-square-error (MMSE). However, the detailed method to estimate these CSIs is beyond the scope of this paper. For notational simplicity, we define $\Omega_l$ by a set of all possible CSI errors associated with the $l$-th PR. In addition, we assume that $\Delta\mathbf{f}_l$ are deterministic and bounded, and therefore $ \delta_l$ represents the size of the uncertainty region of the estimated CSI for the  $l$-th PR.

In addition, a passive Eve does not allow legitimate users to instantaneously obtain  its CSI \cite{Zhou, Nguyen, Li}, which can be justified as the following two reasons. First, to wiretap the confidential messages from both systems, the eavesdroppers require to become as a part of the communication system, i.e., knowing the channel in the downlink. Second, to wiretap a downlink channel without being removed from the system, an eavesdropper has to protect its visibility from the ST without exposing its CSI, for example, not responding its calls (like a passive user). 
 For the passive Eves, we further assume that the entries of $g_{k_p}$,  $\mathbf{f}_{k_p},\;\forall k_p$, $f_{k_g}$, and  $\mathbf{g}_{k_g},\;\forall k_g$, follow independent and identically distributed  (i.i.d.) Rayleigh fading, and that  the instantaneous CSI of these wiretap channels is not available at ST. These assumptions of passive Eves are commonly used in the literature \cite{Nguyen, Yang_14, Zhou, Li}. Meanwhile, the channels $\mathbf{h}_{m_g}, \forall m, g,$ are  assumed to be perfectly known since the SRs are active users in the secondary system.

\subsection{Optimization Problem Formulation}
 Based on the above setting and similar to \eqref{eq:rew:2}, the optimization problem  $\mathbf{P.1}$ can be reformulated as
\begin{IEEEeqnarray}{rCl}\label{eq:imcsi:1}
&&\mathbf{P.2}:\  \underset{\mathbf{w}, \mathbf{U}, \boldsymbol{t}, z, \varphi}{\mathrm{\mathrm{maximize}}}\quad  \varphi\IEEEyessubnumber\label{eq:imcsi:a}\\
   \st &&\   \log_2\bigl(1+\Gamma_{s,m_g}(\mathbf{w},\mathbf{U})\bigl) - t_{g} \geq  \varphi, m_g\in\mathcal{S}_{g},g\in\mathcal{G} \IEEEyessubnumber\label{eq:imcsi:b} \\
	&&  \max_{\mathbf{g}_{k_g}, f_{k_g}}\log_2\bigl(1+\Gamma_{e,k_g}(\mathbf{w},\mathbf{U})\bigl) \leq t_g, k_g\in\mathcal{K}_{e,g}, g\in\mathcal{G} \IEEEyessubnumber\label{eq:imcsi:c} \qquad\\         
	&& \ \min_{\Delta\mathbf{f}_l\in\Omega_l}\log_2\bigl(1+\Gamma_{p,l}(\mathbf{w},\mathbf{U})\bigl) - z\geq \bar{R}_{p,l},\ l\in\mathcal{L} \IEEEyessubnumber\label{eq:imcsi:d}\\
	&&  \ \max_{\mathbf{g}_{k_p}, f_{k_p}}\log_2\bigl(1+\Gamma_{e,k_p}(\mathbf{w},\mathbf{U})\bigl) \leq z,\ k_p\in\mathcal{K}_p\IEEEyessubnumber\label{eq:imcsi:e}\\
     		  &&\quad  \eqref{eq:8c} \IEEEyessubnumber\label{eq:imcsi:f}
\end{IEEEeqnarray}
{\color{black}where  $\boldsymbol{t}\triangleq \{t_{g}\}$ and $z$ are the maximum allowable rates for Eves in decoding the information signals from the ST and the PT, respectively, which were defined in \eqref{eq:rew:1}; $\varphi$ is objective variable to maximize the secrecy rate of the secondary system, which was also defined in \eqref{eq:rew:2}}. Observe that \eqref{eq:imcsi:b} is well presented in \eqref{eq:rew:2:b:equi}. It is now clear that the difficulty in solving \eqref{eq:imcsi:1} is due to \eqref{eq:imcsi:c}-\eqref{eq:imcsi:e} since the remaining constraints are convex and approximate convex. Instead of this, we can find a sub-optimal solution
of \eqref{eq:imcsi:1} as follows
\begin{IEEEeqnarray}{rCl}\label{eq:imcsi:2}
&&\underset{\mathbf{w}, \mathbf{U}, \boldsymbol{t}, z, \varphi, \boldsymbol{\phi},\boldsymbol{\alpha}, \beta}{\mathrm{\mathrm{maximize}}}\quad  \varphi\IEEEyessubnumber\label{eq:imcsi:2:a}\\
   \st &&\quad   
	\log_2\bigl(1+\phi_g\bigl) \leq t_g,\  g\in\mathcal{G} \IEEEyessubnumber\label{eq:imcsi:2:b} \\   
			&&\quad	\Pr\Bigl(\underset{k_g\in\mathcal{K}_{e,g}}{\max}\ \Gamma_{e,k_g}(\mathbf{w},\mathbf{U})\leq \phi_g\Bigl) \geq \epsilon_g ,\  g\in\mathcal{G} \IEEEyessubnumber\label{eq:imcsi:2:c} \\   
	&& \quad\log_2\bigl(1+\alpha_l\bigl) - z\geq \bar{R}_{p,l},\ l\in\mathcal{L} \IEEEyessubnumber\label{eq:imcsi:2:d}\\
	&& \quad\min_{\Delta\mathbf{f}_l\in\Omega_l}\Gamma_{p,l}(\mathbf{w},\mathbf{U}) \geq \alpha_l,\ l\in\mathcal{L} \IEEEyessubnumber\label{eq:imcsi:2:e}\\
	&&  \quad \log_2\bigl(1+\beta\bigl) \leq z\IEEEyessubnumber\label{eq:imcsi:2:f}\\
	&&  \quad \Pr\Bigl(\underset{k_p\in\mathcal{K}_p}{\max}\Gamma_{e,k_p}(\mathbf{w},\mathbf{U}) \leq \beta \Bigl) \geq\tilde{\epsilon} \IEEEyessubnumber\label{eq:imcsi:2:g}\\
     		  &&\quad  \eqref{eq:8c}, \eqref{eq:imcsi:b} \IEEEyessubnumber\label{eq:imcsi:2:h}
\end{IEEEeqnarray}
where $\boldsymbol{\phi}=\{\phi_g\}$, $\boldsymbol{\alpha}=\{\alpha_l\}$, and $\beta$ are newly introduced variables. The constraint \eqref{eq:imcsi:2:e} is imposed to ensure that for a given CSI error set $\Omega_l$, the minimum received SINR at the $l$-th	PR is larger than the minimum SINR requirement $\alpha_l$. According to \eqref{eq:imcsi:2:c} and \eqref{eq:imcsi:2:g}, the probabilities that the maximum received SINR at the $k_g$-th passive Eve and at the $k_p$-th passive Eve are  less than $\phi_{g} > 0$ and $\beta > 0$ are ensured to be greater than $\epsilon_{g}  $ and $\tilde{\epsilon} $, respectively.  To ensure secure communications of the primary system (secondary system), it is required for  $\tilde{\epsilon}$ $(\epsilon_g)$ to be large enough (close to 1).

\subsection{Proposed Solution}

 We are now in position to expose the hidden convexity of the constraint  of \eqref{eq:imcsi:2:c}, \eqref{eq:imcsi:2:e}, and \eqref{eq:imcsi:2:g}. Since $\mathbf{U}$ does not require a rank-constraint matrix, we introduce $ \widetilde{\mathbf{U}} \triangleq \mathbf{U}\mathbf{U}^H$ to facilitate the optimization problem.   Let us handle the constraint  \eqref{eq:imcsi:2:e} first by rewriting it as 
\begin{equation}\label{eq:imcsi:2:e:1}
\max_{\Delta\mathbf{f}_l\in\Omega_l}\sum\nolimits_{g=1}^G|\mathbf{f}_l^{H}\mathbf{w}_g|^2+ \tr(\mathbf{f}_l^{H} \widetilde{\mathbf{U}}\mathbf{f}_l)+\sigma_l^2 \leq \frac{P_p|h_l|^2}{\alpha_l}, l\in\mathcal{L}.
\end{equation} 
For arbitrary $l$-th PR, \eqref{eq:imcsi:2:e:1} can be shaped to take the following equivalent form
\begin{IEEEeqnarray}{rCl}\label{eq:imcsi:2:e:2}
\sum_{g=1}^G\mu_{l,g} + \tilde{\mu}_l +\sigma_l^2 &\leq& \frac{P_p|h_l|^2}{\alpha_l}\label{eq:imcsi:2:e:2a}, l\in\mathcal{L}\\
\max_{\Delta\mathbf{f}_l\in\Omega_l}|\mathbf{f}_l^{H}\mathbf{w}_g|^2 &\leq&  \mu_{l,g},  l\in\mathcal{L}, g\in\mathcal{G}\label{eq:imcsi:2:e:2b}\\
\max_{\Delta\mathbf{f}_l\in\Omega_l} \tr(\mathbf{f}_l^{H} \widetilde{\mathbf{U}}\mathbf{f}_l) &\leq & \tilde{\mu}_l, l\in\mathcal{L}\label{eq:imcsi:2:e:2c}
\end{IEEEeqnarray} 
where $\boldsymbol{\mu}_l=\{\mu_{l,g}\}$ and $\boldsymbol{\tilde{\mu}}=\{\tilde{\mu}_l\}$ are new variables. Note that both sides of \eqref{eq:imcsi:2:e:2a} are convex, so it is iteratively replaced by the following linear constraint
\begin{IEEEeqnarray}{rCl}
\sum_{g=1}^G\mu_{l,g} + \tilde{\mu}_l +\sigma_l^2 &\leq& \frac{2P_p|h_l|^2}{\alpha_l^{(n)}} - \frac{P_p|h_l|^2}{(\alpha_l^{(n)})^2}\alpha_l,\label{eq:imcsi:2:e:2a1}l\in\mathcal{L}.
\end{IEEEeqnarray} 
To make the tractable form of \eqref{eq:imcsi:2:e:2b} and \eqref{eq:imcsi:2:e:2c}, we first transform these constraints into a  matrix inequality  based on the following lemma.
\begin{lemma}\label{lemma:Spro}
\textsl{(S-Procedure} \cite{Stephen}): Let $f_m(\mathbf{x}) = \mathbf{x}^{H}\mathbf{A}_m\mathbf{x} + 2\mbox{Re}\{\mathbf{b}^{H}_m\mathbf{x}\} + c_m$, where $m=\{1,2\}$, $\mathbf{A}_m\in\mathbb{H}^{N}, \mathbf{b}_m\in\mathbb{C}^{N\times 1}$ and $c_m\in\mathbb{R}$. Then there exists a $\mathbf{\hat{x}}$ such that $f_z(\mathbf{\hat{x}}) < 0$ satisfies: $f_1(\mathbf{x})\leq 0$ $\Rightarrow$ $f_2(\mathbf{x})\leq 0$ if and only if there exists  $\omega\geq 0$ such that
\begin{equation}
\omega\begin{bmatrix}
    \mathbf{A}_1       & \mathbf{b}_1\\
    \mathbf{b}_1^{H}       & c_1\\
\end{bmatrix}
-
\begin{bmatrix}
    \mathbf{A}_2 & \mathbf{b}_2\\
    \mathbf{b}_2^{H} & c_2\\
 \end{bmatrix}\succeq\mathbf{0}.
\end{equation}
 \end{lemma}
Substituting $\mathbf{f}_l=\mathbf{\hat{f}}_l+\Delta\mathbf{f}_l, \forall l$ into \eqref{eq:imcsi:2:e:2b} and applying Lemma \ref{lemma:Spro}, then        
\begin{equation}
\begin{aligned}
            &\Delta\mathbf{f}_l^{H}\Delta\mathbf{f}_l-\delta_l^2\leq0\\
\Rightarrow \eqref{eq:imcsi:2:e:2b}:\;&\Delta\mathbf{f}_l^{H}\mathbf{w}_g\mathbf{w}_g^H\Delta\mathbf{f}_l+2\Re\{\mathbf{\hat{f}}_l^{H}\mathbf{w}_g\mathbf{w}_g^H\Delta\mathbf{f}_l\}  \\
&+\mathbf{\hat{f}}_l^{H}\mathbf{w}_g\mathbf{w}_g^H\mathbf{\hat{f}}_l-\mu_{l,g}\leq 0
\end{aligned}\end{equation}
holds if and only if there exists $\boldsymbol{\omega}_{l}=\{\omega_{l,g}\geq 0\},\;\forall l$, so that the following matrix inequality constraint holds  
\begin{equation}\begin{aligned}
\begin{bmatrix}
    \omega_{l,g}\mathbf{I}_N-\mathbf{w}_g\mathbf{w}_g^H & -\mathbf{w}_g\mathbf{w}_g^H\mathbf{\hat{f}}_l\\
    -\mathbf{\hat{f}}_l^{H}\mathbf{w}_g\mathbf{w}_g^H & -\mathbf{\hat{f}}_l^{H}\mathbf{w}_g\mathbf{w}_g^H\mathbf{\hat{f}}_l-\omega_{l,g}\delta_l^2+\mu_{l,g}
 \end{bmatrix}
\succeq\mathbf{0}.
\end{aligned}\label{eq:LMI:constraint:2}\end{equation}
However, \eqref{eq:LMI:constraint:2} is still not in a tractable form. At this point, we apply the application of Schur's complement lemma \cite[Eq.~(7.2.6)]{Tal:09} to obtain the following linear matrix inequality (LMI)
\begin{equation}\begin{aligned}
&\exists \omega_{l,g} \geq 0: \mathbf{C}_{l,g}(\mathbf{w}_g,\mu_{l,g},\omega_{l,g})\triangleq \\
&\begin{bmatrix}
    1 &     \mathbf{w}_g^H      &-\mathbf{w}_g^H\mathbf{\hat{f}}_l\\
		\mathbf{w}_g                                                    &  \omega_{l,g}\mathbf{I}_N  &                                                    \\
    -\mathbf{\hat{f}}_l^{H}\mathbf{w}_g & &-\omega_{l,g}\delta_l^2+\mu_{l,g}
 \end{bmatrix}
\succeq\mathbf{0},\ g\in\mathcal{G}, l\in\mathcal{L}.
\end{aligned}\label{eq:LMI:constraint:3}\end{equation}
It is also worth noting that constraint \eqref{eq:LMI:constraint:3} now includes only a finite number of constraints. 

Analogously, with $\tilde{\boldsymbol{\omega}}=\{\tilde{\omega}_{l}\geq 0\},\;$ the constraint \eqref{eq:imcsi:2:e:2c} admits the following representation
\begin{equation}\begin{aligned}
&\exists \tilde{\omega}_l \geq 0: \tilde{\mathbf{C}}_{l}(\widetilde{\mathbf{U}},\tilde{\mu}_{l},\tilde{\omega}_{l})\triangleq\\
&\begin{bmatrix}
    \tilde{\omega}_{l}\mathbf{I}_N - \widetilde{\mathbf{U}} &- \widetilde{\mathbf{U}}\mathbf{\hat{f}}_l\\
		-\mathbf{\hat{f}}_l^{H} \widetilde{\mathbf{U}} &-\mathbf{\hat{f}}_l^{H} \widetilde{\mathbf{U}}\mathbf{\hat{f}}_l-\tilde{\omega}_{l}\delta_l^2+\tilde{\mu}_{l}
 \end{bmatrix}
\succeq\mathbf{0},\ l\in\mathcal{L}.
\end{aligned}\label{eq:LMI:constraint:4}\end{equation}

To deal with the nonconvex constraints given in  \eqref{eq:imcsi:2:g} and \eqref{eq:imcsi:2:c}, we provide the following two lemmas. 
\begin{lemma}\label{lemma:PR}
For the primary system, the constraint  \eqref{eq:imcsi:2:g} is transformed to a new constraint as
\begin{equation}
\lambda_{\min}\left(\sum\nolimits_{g=1}^G\mathbf{w}_g\mathbf{w}_g^H +\widetilde{\mathbf{U}}\right) \geq \tilde{\xi}(\beta)
\label{eq:lemma:2}\end{equation}
where $\tilde{\xi}(\beta) \triangleq  \Bigl(\exp\bigl(-\frac{\beta}{NP_p}\sigma_{k_p}^2\bigr)/(1-\tilde{\epsilon}^{1/K_p})^{1/N}-1\Bigl)\frac{P_p}{\beta}.$
\end{lemma}
\begin{proof}
See Appendix C.
\end{proof}
In  Lemma \ref{lemma:PR}, the claim is clearly true in the trivial case of $\beta \rightarrow \infty$, i.e., the primary system is  inactive, which leads to $\sum_{g=1}^G\mathbf{w}_g\mathbf{w}_g^H +\widetilde{\mathbf{U}}\succeq\mathbf{0}$. This is always true and thus confirms our analysis. Next, we rewrite \eqref{eq:lemma:2} equivalently in the form of
\begin{IEEEeqnarray}{rCl}
2\ln\eta + \beta\frac{\sigma_{k_p}^2}{NP_p} &\geq & 0 \label{eq:lemma2:a}\\
\bigl(\eta^2/(1-\tilde{\epsilon}^{1/K_p})^{1/N} -1\bigr)P_p &\leq& \beta \theta  \label{eq:lemma2:b}\\
\lambda_{\min}\Bigr(\sum\nolimits_{g=1}^G\mathbf{w}_g\mathbf{w}_g^H +\widetilde{\mathbf{U}}\Bigl) &\geq& \theta \label{eq:lemma2:c}
\end{IEEEeqnarray}
where $\theta$ and $\eta$ are newly introduced variables. Since the constraints \eqref{eq:lemma2:a} and \eqref{eq:lemma2:b} are convex, and we now focus on the remaining nonconvex constraint. 
In  \eqref{eq:lemma2:c}, we note that both $\sum\nolimits_{g=1}^G\mathbf{w}_g\mathbf{w}_g^H$ and $\widetilde{\mathbf{U}}$ are Hermitian matrices. In addition,  the eigenvalues of a Hermitian matrix $\mathbf{Q}$ are real and satisfy $\tr(\mathbf{x}^H\mathbf{Q}^H\mathbf{x}) \ge \lambda\|\mathbf{x}\|^2$ for any given vector $\mathbf{x}$ if and only if $\lambda_{\min}(\mathbf{Q}) \geq \lambda$. Since $\lambda_{\min}(\mathbf{w}_g\mathbf{w}_g^H)=0$ for all $g$, the lower bound of left side of \eqref{eq:lemma2:c} is given by
\begin{IEEEeqnarray}{rCl}
\lambda_{\min}\Bigr(\sum\nolimits_{g=1}^G\mathbf{w}_g\mathbf{w}_g^H +\widetilde{\mathbf{U}}\Bigl)\geq \lambda_{\min}(\widetilde{\mathbf{U}}). \label{eq:lemma2:c1}
\end{IEEEeqnarray}
 The implication of \eqref{eq:lemma2:c1} is that the ST will degrade  the eavesdropper's channel by transmitting jamming noise rather than the desired signals. From \eqref{eq:lemma2:c}, it follows that 
\begin{IEEEeqnarray}{rCl}
 \lambda_{\min}(\widetilde{\mathbf{U}}) \geq \theta \Leftrightarrow \widetilde{\mathbf{U}} \succeq \mathbf{I}_N\theta \label{eq:lemma2:c2}.
\end{IEEEeqnarray}

\begin{lemma}\label{lemma:SR}
For the secondary system, the constraint  \eqref{eq:imcsi:2:c}  is transformed to  a new constraint as
\begin{equation}
\frac{\|\mathbf{w}_g\|^2}{\phi_g} \leq \xi_g  + \sum_{i=1,i\neq g}^G\|\mathbf{w}_i\|^2 + \lambda_{\min}(\widetilde{\mathbf{U}}),\ g\in\mathcal{G}
\label{eq:lemma:3}\end{equation}
where $\xi_g\triangleq\Bigl[\exp\Bigl(\frac{\sigma_{k_g}^2}{NP_p}\Bigr)\epsilon_{g}^{-1/NK_g}-1\Bigr]P_p.$
\end{lemma}
\begin{proof}
See Appendix D.
\end{proof}
The formulation in  \eqref{eq:lemma:3} can be further shaped to take the following convex constraints
\begin{IEEEeqnarray}{rCl}
&&\frac{\|\mathbf{w}_g\|^2}{\phi_g} \leq \xi_g  + \sum_{i=1,i\neq g}^G 2\Re\{(\mathbf{w}_i^{(n)})^H\mathbf{w}_i\} \nonumber\\
&&\qquad\qquad\qquad- \sum_{i=1,i\neq g}^G\|\mathbf{w}_i^{(n)}\|^2 + \vartheta,\ g\in\mathcal{G} \label{eq:lemma:3a}\\
&&\lambda_{\min}(\widetilde{\mathbf{U}}) \geq \vartheta \Leftrightarrow  \widetilde{\mathbf{U}} \succeq \mathbf{I}_N \vartheta \label{eq:lemma:3c}
\end{IEEEeqnarray}
where $\vartheta$  is newly introduced variable.
\begin{remark}
We note that the new constraints in \eqref{eq:lemma:2} and \eqref{eq:lemma:3} are not equivalent to \eqref{eq:imcsi:2:g} and \eqref{eq:imcsi:2:c}. Specifically, the optimal solutions for the former are also feasible for the latter, respectively, but not vice versa due to the inequalities in \eqref{eq:proba:pri:2} and \eqref{eq:proba:second:2}, and thus this leads to a lower bound of the system performance.
\end{remark}
{\color{black}\begin{remark}
In this paper, the wiretap channels are modeled as i.i.d. Rayleigh random variables. Nevertheless, a different continuous channel distribution does not affect the type of  constraints in \eqref{eq:lemma:2} and \eqref{eq:lemma:3}. In other words, the proposed convex approximation is still applicable to any continuous channel distribution thanks to widespread applications of inner approximation method \cite{Marks:78}. Therefore, our study is valid without loss of generality.
\end{remark}}

With the above discussions, the approximate convex problem solved at the $(n+1)$-th iteration of the  proposed design is given by
\begin{IEEEeqnarray}{rCl}\label{eq:imcsi:3}
&&\underset{\substack{\mathbf{w}, \widetilde{\mathbf{U}}\succeq\mathbf{0}, \boldsymbol{t}, z, \varphi, \boldsymbol{\phi},\boldsymbol{\alpha},\\ \beta, \boldsymbol{\mu}_l,\tilde{\boldsymbol{\mu}}, \boldsymbol{\omega}_l,\tilde{\boldsymbol{\omega}},\theta, \eta, \vartheta}}{\mathrm{\mathrm{maximize}}}\quad  \varphi \IEEEyessubnumber\\
   && \st \quad   
	 \mathcal{F}_{m_g}^{(n)}(\mathbf{w},\widetilde{\mathbf{U}}) \geq (\varphi + t_g)\ln2,\ m_g\in\mathcal{S}_{g},g\in\mathcal{G} \IEEEyessubnumber \qquad\\
     		  &&\quad\quad  \sum\nolimits_{g=1}^G\|\mathbf{w}_g\|^2+ \tr(\widetilde{\mathbf{U}}) \leq P_{s}\, \IEEEyessubnumber\label{eq:8c:revised} \\
					&&\quad\quad   \eqref{eq:imcsi:2:b}, \eqref{eq:imcsi:2:d}, \eqref{eq:imcsi:2:f}, \eqref{eq:imcsi:2:e:2a1}, \eqref{eq:LMI:constraint:3},  \nonumber \\
					&&\quad\quad \eqref{eq:LMI:constraint:4}, \eqref{eq:lemma2:a}, \eqref{eq:lemma2:b}, \eqref{eq:lemma2:c2}, \eqref{eq:lemma:3a}, \eqref{eq:lemma:3c}. \IEEEyessubnumber\label{eq:imcsi:3:h}
\end{IEEEeqnarray}
To find an initial feasible point to \eqref{eq:imcsi:1}, we solve the following convex optimization problem
\begin{IEEEeqnarray}{rCl}\label{ipcsi:ini3.m}
&&\max_{\substack{\mathbf{w}, \widetilde{\mathbf{U}}\succeq\mathbf{0},  z,\boldsymbol{\alpha}, \beta,\\ \boldsymbol{\mu}_l,\tilde{\boldsymbol{\mu}}, \boldsymbol{\omega}_l,\tilde{\boldsymbol{\omega}},\theta,\eta }}\;\min_{l\in\mathcal{L}}\;
\Bigl\{\log_2\bigl(1+\alpha_l\bigl) - z - \bar{R}_{p,l}
  \Bigr\}\IEEEyessubnumber \label{eq:ipcsi:ini3:a}\\
	&&\quad\st\ \,  \eqref{eq:imcsi:2:f}, \eqref{eq:imcsi:2:e:2a1}, \eqref{eq:LMI:constraint:3}, \eqref{eq:LMI:constraint:4},\eqref{eq:lemma2:a}, \eqref{eq:lemma2:b}, \eqref{eq:lemma2:c2}, \eqref{eq:8c:revised} \IEEEyessubnumber \label{eq:ipcsi:ini3:c}
\end{IEEEeqnarray}
and  stop at reaching
\begin{equation}\label{ipcsi:ini4.m}
\min_{l\in\mathcal{L}}\;
\Bigl\{\log_2\bigl(1+\alpha_l\bigl) - z - \bar{R}_{p,l}
  \Bigr\}\geq 0.
\end{equation}
The proposed iterative method is outlined in Algorithm \ref{algo:proposed:DUAL:2}. In a similar manner to Proposition \ref{prop1}, we can show  that Algorithm \ref{algo:proposed:DUAL:2} yields a nondecreasing  sequence of  objective  due to  updating the involved variables after each iteration.
\begin{algorithm}[t]
\begin{algorithmic}[1]

\protect\caption{An iterative algorithm to solve \eqref{eq:imcsi:1}}

\label{algo:proposed:DUAL:2}

\global\long\def\algorithmicrequire{\textbf{Initialization:}}

\REQUIRE  Set $n:=0$ and solve \eqref{ipcsi:ini3.m} to generate an  initial feasible point $\bigl(\mathbf{w}^{(n)},\widetilde{\mathbf{U}}^{(n)}, \boldsymbol{\alpha}^{(n)}\bigr)$ 

\REPEAT
\STATE Solve \eqref{eq:imcsi:3} to obtain the optimal solution: $\bigl(\mathbf{w}^{*},\widetilde{\mathbf{U}}^{*}, \boldsymbol{\alpha}^{*})$.

\STATE Update\ $\mathbf{w}^{(n+1)}:=\mathbf{w}^{*}$,  $\widetilde{\mathbf{U}}^{(n+1)}:=\widetilde{\mathbf{U}}^{*}$, and $\boldsymbol{\alpha}^{(n+1)} := \boldsymbol{\alpha}^{*}$.

\STATE Set $n:=n+1.$
\UNTIL Convergence or maximum required number of iterations\\
\end{algorithmic} \end{algorithm}

\textit{Complexity Analysis}: The optimization problem in \eqref{eq:imcsi:3}  involves $GL$ LMI constraints of size $N+2$, $L$ LMI constraints of size $N+1$, and 2 LMI constraints of size $N$. Since the major complexity of solving  \eqref{eq:imcsi:3} comes from LMI constraints, we ignore the complexity of the constraints of lower sizes and they will not affect the complexity order of the whole problem. As a result, in each iteration of Algorithm \ref{algo:proposed:DUAL:2}, the worst-case computational complexity for solving the generic convex problem in \eqref{eq:imcsi:3} using interior point methods is given by $\mathcal{O}\Bigl(n\sqrt{GL(N+2) + L(N+1) + 2N}\bigl[GL(N+2)^3 + L(N+1)^3 + 2N^3 + nGL(N+2)^2 + nL(N+1)^2 + 2nN^2 +n^2\bigl]\Bigl)$, where $n = G(L+3)+N(N+G)+2L+6$ \cite{Ben:2001}.

\section{Numerical Results And Discussions}\label{Numerical}
In this section, we use simulations to evaluate the performance of the proposed approach.  The number of groups of SUs is set to $G=2$,  each of which consists of two SR users, i.e., $M_g=2,\,\forall g$. The number of  PRs is set to $L=2$, and each group of SUs and PUs is surrounded by two Eves, i.e., $K_p = K_g = 2$. All  channel entries are assumed to be i.i.d. complex Gaussian random variables with $\mathcal{CN}(0,1)$, and the background thermal noise at each user is generated as i.i.d. complex Gaussian random variables  with zero means and unit variance. The transmit power at the PT is fixed to $P_p=20$ dBm. For simplicity, we further assume that the minimum secrecy rate requirement for all PUs are the same, i.e., $\bar{R}_{p,l}=\bar{R}_{p}, \forall l$.  For the imperfect CSI of the PU channels, we define the normalized channel estimation errors as $\bar{\delta}^2_l=\delta^2_l/\|\mathbf{f}_l\|^2 =5\%$, $\forall l$. To guarantee secure communications, we choose $\tilde{\epsilon}=0.99$ and $\epsilon_g=0.99,\,\forall g$ for the passive Eves.
The results obtained in this paper are referred to as the proposed optimal scheme.
We also compare the performance of the proposed scheme with the known solutions, namely, the ``No JN scheme'' \cite{Yiyang,Pei} and ``Partial ZF (zero-forcing) scheme'' \cite{Nguyen:TIFS:16}. In the ``No JN scheme,''  the optimal solution   can  be obtained  by setting $\mathbf{U}$  to $\mathbf{0}$. In the ``Partial ZF scheme,'' we consider the  null space approach at the ST.  First of all, the JN is transmitted to all  Eves and to avoid interfering with both  PUs and SUs  as
\begin{equation}
\mathbf{U}^{H}\mathbf{f}_l=0, \forall l\quad \mbox{and}\quad  \mathbf{U}^{H}\mathbf{h}_{m_g}=0, \forall m_g, g.
\label{eq:zf1}
\end{equation}
In a CRN, the primary system should have higher priority, and thus the transmitted information at the ST should not generate interferences to the PUs  as
\begin{equation}
\mathbf{w}^{H}_g\mathbf{f}_l=0, \forall  l, g.
\label{eq:zf2}
\end{equation}
To simplify the problem, we enforce the information transmitted at the ST so that it should not introduce interference to  other groups as
\begin{equation}
\quad\mathbf{w}^{H}_g\mathbf{h}_{m_i}=0, \forall  i\neq g.
\label{eq:zf3}
\end{equation}
It is evident that $\Gamma_{p,l}=\frac{P_p|h_l|^2}{\sigma^2_l},  \forall l,$ does not depend on $\mathbf{w}_g$ and $\mathbf{U}$. So, we utilize \eqref{eq:zf1}, \eqref{eq:zf2}, and \eqref{eq:zf3} into $\mathbf{P1}$ to obtain the optimal solution for ``Partial ZF scheme.''  To solve convex problems we 
use the SDPT3  as the internal solver \cite{Toh} in MATLAB environment. The results of the secrecy rate  are shown by averaging over 1,000 simulation trials.

\begin{figure}
    \begin{center}
    \begin{subfigure}[Convergence results of Algorithm \ref{algo:proposed:DUAL} for different numbers of antennas at the ST.]{
        \includegraphics[width=0.46\textwidth]{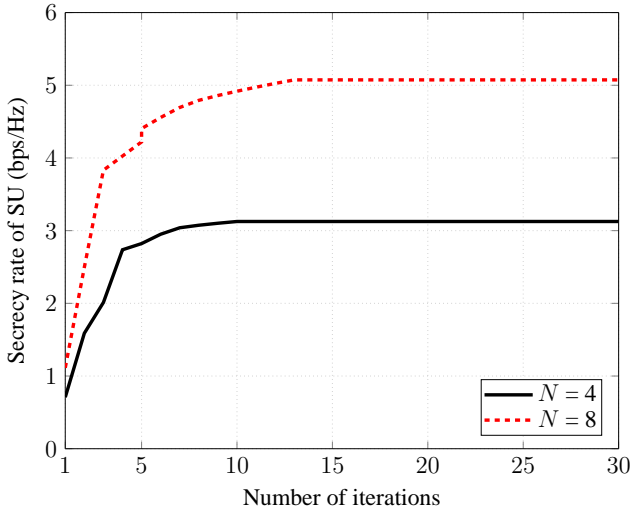}}
    		\label{fig:a}
				\end{subfigure}
				 \begin{subfigure}[Convergence results of Algorithm \ref{algo:proposed:DUAL:2} for different numbers of antennas at the ST.]{
        \includegraphics[width=0.46\textwidth]{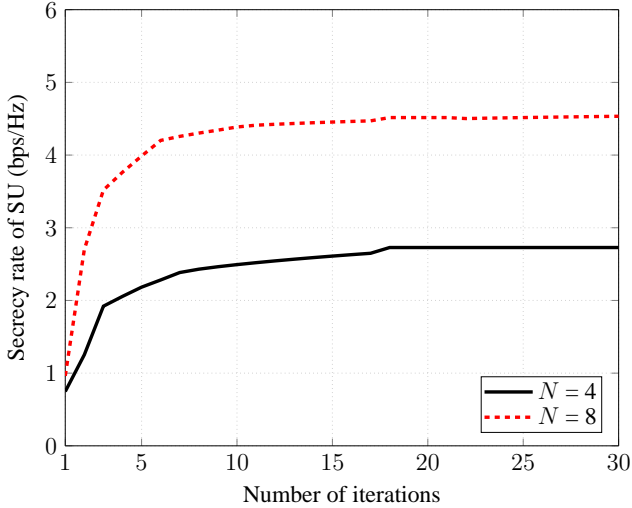}}
        \label{fig:b}
    \end{subfigure}
	  \caption{Convergence results of Algorithm \ref{algo:proposed:DUAL} and \ref{algo:proposed:DUAL:2} for different numbers of antennas at the ST over one random channel realization with $\bar{R}_{p} = 2$ bps/Hz and $P_s$ = 15 dBm.}\label{fig:Convergencebehavior:Iteration}
\end{center}
\end{figure}

Fig.~\ref{fig:Convergencebehavior:Iteration} illustrates the typical convergence behavior of the proposed Algorithm \ref{algo:proposed:DUAL}
 and Algorithm \ref{algo:proposed:DUAL:2}  as a function of the number of iterations with different numbers of antennas at the ST for Algorithm \ref{algo:proposed:DUAL} in Fig.~\ref{fig:Convergencebehavior:Iteration}(a) and for Algorithm \ref{algo:proposed:DUAL:2} in Fig.~\ref{fig:Convergencebehavior:Iteration}(b). As seen, the objective values of both algorithms increase rapidly within the first 10 iterations and stabilize  after a few more iterations, and its convergence rate is slightly sensitive to the problem size i.e., as $N$  increases. The convergence results also confirm that all optimization variables are accounted  to find a better solution for the next iteration, i.e., the secrecy rates of SUs  monotonically increasing. In addition, Fig.~\ref{fig:Convergencebehavior:Iteration} shows that at least $90\%$ of  secrecy rate is obtained when the proposed algorithms reach to 10 iterations. 

\begin{figure}[t]
\centering
\includegraphics[ width=0.48\textwidth]{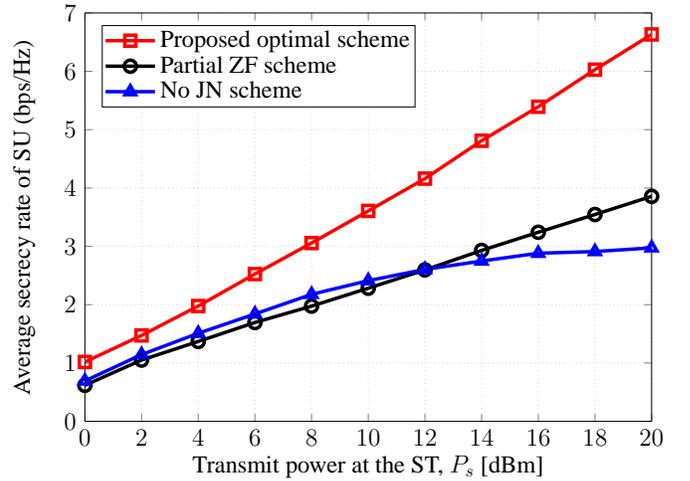}
\caption{ Average secrecy rate of the secondary system  vs. the transmit power at the ST with perfect CSI, where $\bar{R}_{p} = 2$ bps/Hz and $N = 8$.}
\label{fig:Fig3:Ps}
\end{figure}

Fig.~\ref{fig:Fig3:Ps} plots the average secrecy rate of  secondary
system versus the transmit power at the ST. As can be seen, the proposed optimal
scheme  greatly improves the secrecy rate of the ``Partial ZF scheme'' and ``No JN scheme,'' especially in high power regime.
The performance gain is thus achieved as a result of more intelligent interference management than that of other 
schemes for primary users and  Eves. Another interesting observation is that the ``No JN scheme'' outperforms the ``Partial
ZF scheme'' in low power regime ($P_s \leq 12$ dBm), but it tends to saturate when the transmit power becomes high. This is 
mainly due to the fact that, in high power regime, the ST needs to scale down the transmit power to maintain the secrecy rate of the primary
system, which results in a loss of the secrecy rate of  secondary system. Moreover, the simulation results in Fig.~\ref{fig:Fig3:Ps}
further confirm that incorporating  information and JN beamforming is a powerful means to transmit with  full power.

\begin{figure}[t]
\centering
\includegraphics[ width=0.48\textwidth]{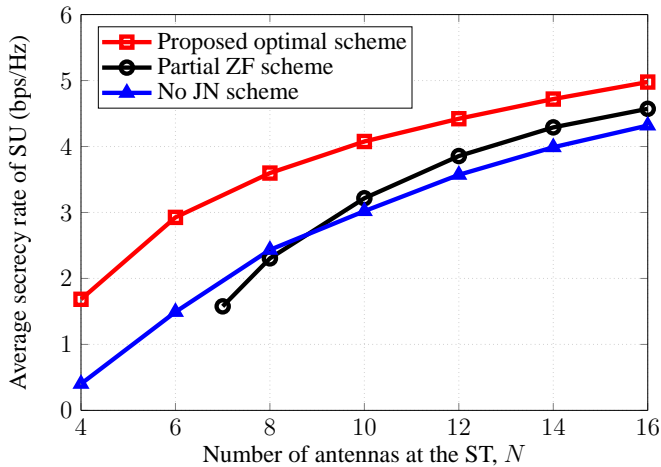}
\caption{ Average secrecy rate of the secondary system  vs. the number of transmit antennas at the ST with  perfect CSI,  where  $\bar{R}_{p} = 2$ bps/Hz and $P_s$ = 10 dBm.}
\label{fig:Fig4:N}
\end{figure}

In Fig.~\ref{fig:Fig4:N}, we study the secrecy rate of  secondary system as a function of the number of transmit antennas at the ST, $N$. The results show that the achievable secrecy rate increases as the number of transmit antennas increases in all schemes, since more degrees of freedom are added to the ST.  The proposed optimal scheme still achieves a better performance than  other schemes in all the range of $N$. We note that the optimal solution for the ``Partial ZF scheme'' is infeasible when $N < 7$ because for the ``Partial ZF scheme,''  interference among  legitimate users cannot be completely canceled out with insufficient number of transmit antennas. As expected, the gap between the proposed scheme and ``Partial ZF scheme'' is reduced as a result of providing more degrees of freedom.

\begin{figure}
    \centering
    \begin{subfigure}[Average secrecy rate of the secondary system for  different schemes, where ${P_s}$ = 10 dBm.]{
        \includegraphics[width=0.46\textwidth]{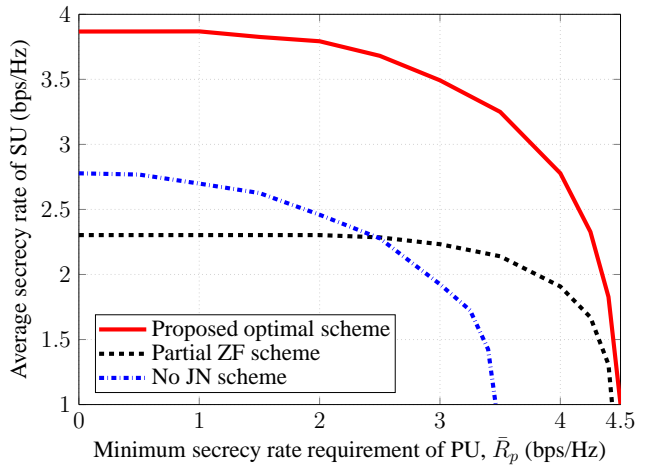}}
    		\label{fig:a:PU}
				\end{subfigure}
		\hfill
    \begin{subfigure}[Average secrecy rate of the secondary system for  different power sharing, where ${P_s}$ = 20 dBm]{
        \includegraphics[width=0.46\textwidth]{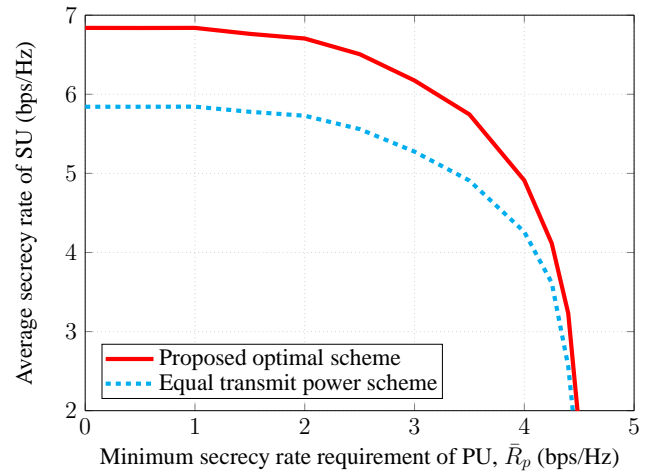}}
        \label{fig:b:PU}
    \end{subfigure}
	  \caption{Average secrecy rate of the secondary system vs. the minimum secrecy rate requirement of the primary system, (a) for  different schemes and (b) for  different  power sharing with perfect CSI, where $N$ = 8.}\label{fig:Serecyrate_vs_PU}
\end{figure}

The average secrecy rate of the secondary system is investigated as a function of  the minimum secrecy rate requirement of  primary system, $\bar{R}_p$, in Fig.~\ref{fig:Serecyrate_vs_PU}(a) for  different schemes and in Fig.~\ref{fig:Serecyrate_vs_PU}(b) for different  power sharing. As can be seen from Fig.~\ref{fig:Serecyrate_vs_PU}(a), the secondary system achieves a higher secrecy rate with the proposed optimal scheme than with  other schemes. Notably, the performance of  ``No JN scheme'' is degraded significant as $\bar{R}_p$ increases. The main reason for such a case is that, the ST is required to cause  less interference to the PRs and transmit  high interference to degrade the Eves' channels, which results in  a significant loss of the secondary system's secrecy rate. The secrecy rate of the ``Partial ZF scheme'' is nearly unchanged when $\bar{R}_p$ increases and approaches that of the proposed optimal scheme for high  $\bar{R}_p$, since the ST does not cause any interference to the PRs. In Fig.~\ref{fig:Serecyrate_vs_PU}(b), we plot the average secrecy rate of the secondary system  for the proposed optimal scheme under  different assumption of sharing equally the resources, i.e., transmit power at the ST. Particularly,  the information and JN beamforming are assumed to share 50$\%$ of the power resource, i.e., $\sum\nolimits_{g=1}^G\|\mathbf{w}_g\|^2\leq P_s/2$ and  $\|\mathbf{U}\|^2\leq P_s/2$. As seen, the proposed joint information and JN beamforming offers better performance compared to that of the equal transmit power scheme. However, the gap between the schemes diminishes for high secrecy rate of the primary system. {\color{black}The reason for this is two-fold: 1) For small $\bar{R}_p$, a small portion of JN already fulfills the QoS requirement of PU, and there is no need to further waste power budget on JN; 2) For extremely stringent QoS requirement of PU, JN becomes crucial and so it is reasonable to allocate a significant part of the power budget to JN (i.e., nearly a half as shown in Fig.~\ref{fig:Serecyrate_vs_PU}(b)) to meet the QoS requirement}. From both Figs.~\ref{fig:Serecyrate_vs_PU}(a) and~\ref{fig:Serecyrate_vs_PU}(b), for high $\bar{R}_p$, the secondary system lacks degree of freedom for leveraging multiuser diversity.

\begin{figure}[t]
\centering
\includegraphics[ width=0.48\textwidth]{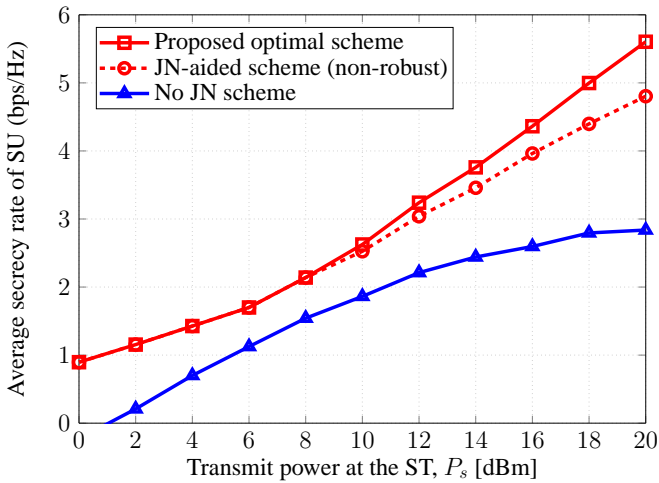}
\caption{ Average secrecy rate of the secondary system  vs. the transmit power at the ST with realistic scenario, where $\bar{R}_{p} = 1$ bps/Hz and $N = 8$.}
\label{fig:Fig6:Ps:IPCSI}
\end{figure}

We now turn our attention to illustrate the robustness of the proposed  design in  realistic scenario.
We also compare the performance of the proposed robust design to that of  non-robust secrecy rate. For the non-robust secrecy rate design, we use the presumed CSIs as $\hat{\mathbf{f}}_l,\ \forall l$ rather than the true ones, to perform the transmit design (as presented in Section IV), which then evaluates the resultant  secrecy rate. Fig.~\ref{fig:Fig6:Ps:IPCSI} depicts the  secrecy rate  as a function of the transmit power at the ST. As can be observed that the secrecy rate of non-robust design  is sensitive to the CSI uncertainties for high $P_s$. In particular, when $P_{s} \geq 8$ dBm, the non-robust design exhibits the degradation in terms of the secrecy rate that tends to worsen as $P_{s} $ increases. Moreover, the proposed optimal design achieves the best secrecy rate performance, compared to  other designs.

\begin{figure}
    \centering
    \begin{subfigure}[CDF of the secrecy rate of the secondary system for  different schemes, where ${\bar{R}_p}$ = 1 bps/Hz.]{
        \includegraphics[width=0.46\textwidth]{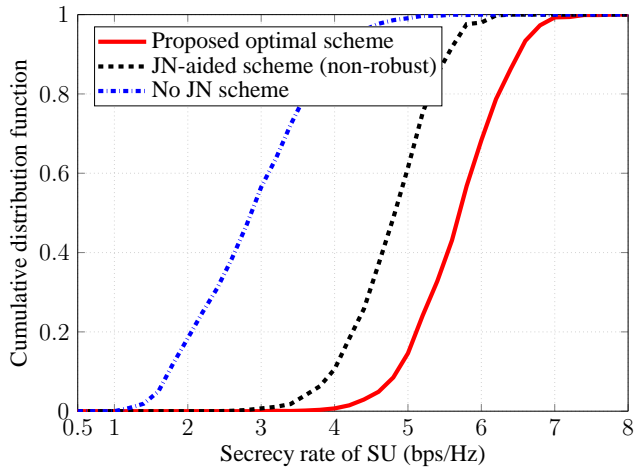}}
    		\label{fig:a:cdf}
				\end{subfigure}
		\hfill
    \begin{subfigure}[CDF of the secrecy rate of the secondary system for  different power sharing.]{
        \includegraphics[width=0.46\textwidth]{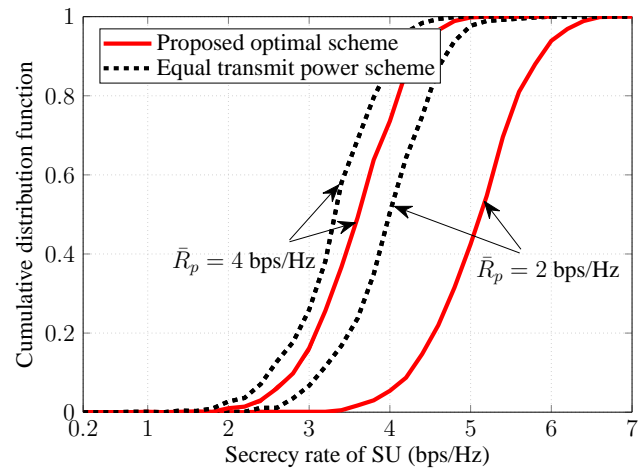}}
        \label{fig:b:cdf}
    \end{subfigure}
	  \caption{CDF of secrecy rate of the secondary system, {\color{black}the probability that the secrecy rate  will take a value less than or equal to a given secrecy rate threshold}, (a) for  different schemes and (b) for  different  power sharing with  realistic scenario, where $N$ = 8 and $P_s$ = 20 dBm.}\label{fig:CDF}
\end{figure}

Finally, we generate cumulative distribution functions (CDFs) of the secrecy rate of the secondary system in Fig.~\ref{fig:CDF}(a) for different schemes and in Fig.~\ref{fig:CDF}(b) for different power sharing. It is obvious in both CDFs that on account for a larger feasible set, the proposed optimal scheme can promise a bigger secrecy rate as expected. For instance, the proposed optimal scheme attains 0.8 bps/Hz and 2.8 bps/Hz of the achievable secrecy rate  higher than the non-robust scheme and ``No JN scheme,'' respectively, for approximately $60\%$ of the simulated trials in Fig.~\ref{fig:CDF}(a). For  large $\bar{R}_p$, the gap between the proposed design and non-robust design is reduced  as in Fig.~\ref{fig:CDF}(b) due to a decrease in the available multiuser diversity gain.

\section{Conclusion}\label{Conclusion}
In this paper, we have considered  PHY security for both  primary and secondary systems in the presence of  multiple secondary receiver groups and multiple primary receivers. The secondary system has been proposed to assist the primary system  by sending  jamming noise to degrade the decoding capability of the
 eavesdroppers. The main objective  is to maximize the secrecy rate of the secondary system, while the secondary transmitter is constrained not only by the power budget, but also by the individual minimum secrecy rate requirements of the primary users.  We have proposed iterative  algorithms to solve the  optimization problems. The idea of the proposed method is to approximate the nonconvex problem by a convex formulation in each iteration. We have proved that our iterative algorithms are guaranteed to monotonically converge to at least local optima of  the original nonconvex design problems. We have carried out  simulations  to evaluate the advantages of the proposed design. It has been shown that for a given initial feasible point, the proposed iterative algorithms are guaranteed to always converge to an optimal solution.

\section*{Appendix~A\\Proof of Lemma~\ref{LemmaPCSI}} \label{Appendix:A1}
The following inequalities play an important role in our developments:
\begin{IEEEeqnarray}{rCl}
&&\ln\Bigl(1+\frac{|x|^2}{y}\Bigr)\geq \ln\Bigl(1+\frac{|x^{(n)}|^2}{y^{(n)}}\Bigr)-\frac{|x^{(n)}|^2}{y^{(n)}}
+2\frac{\Re\{(x^{(n)})^*x\}}{y^{(n)}} \nonumber\\
&&\qquad\qquad\qquad\quad -\; \frac{|x^{(n)}|^2(|x|^2+y)}{y^{(n)}(y^{(n)}+|x^{(n)}|^2)}, \forall x\in\mathbb{C},  y>0,\label{inq1}\\
&&\frac{|x|^2}{y}\geq 2\frac{(x^{(n)})^*x}{y^{(n)}}-\frac{|x^{(n)}|^2}{(y^{(n)})^2}y,\  \forall x\in\mathbb{C},  y>0\label{inq2},\\
&&\ln(1+x) \leq \ln\bigl(1+x^{(n)}\bigr) + \frac{(x-x^{(n)})}{(1+x^{(n)})},\ \forall x\geq 0  \label{inq3}
\end{IEEEeqnarray}
where \eqref{inq1} and \eqref{inq2} follow from the convexity of functions $\ln\bigl(1+|x|^2/y\bigr)$ and $|x|^2/y$ \cite{TuyBook,Ngetal17}, respectively; while \eqref{inq3} is a result of the concavity of function $\ln(1+x)$.

Let us treat the nonconvex constraint \eqref{eq:rew:2:b} first. As the first step, \eqref{eq:SINR:sr} is equivalently rewritten by
\begin{equation}\label{f1sinr1}
\Gamma_{s,m_g}(\mathbf{w},\mathbf{U})=\frac{|\mathbf{h}_{m_g}^{H}\mathbf{w}_g|^2}{\chi_{s,m_g}(\mathbf{w},\mathbf{U})}
\end{equation}
where
\begin{IEEEeqnarray}{rCl}
\chi_{s,m_g}(\mathbf{w},\mathbf{U})=\sum_{i=1, i\neq g }^G|\mathbf{h}_{m_g}^{H}\mathbf{w}_i|^2+ \|\mathbf{h}_{m_g}^{H}\mathbf{U}\|^2 \nonumber\\
+\;  P_p|f_{m_g}|^2+\sigma_{m_g}^2.\nonumber
\end{IEEEeqnarray}
From \eqref{f1sinr1}, it follows that
\begin{equation}\label{eq:srlog:1}
\ln\Bigr(1+\frac{|\mathbf{h}_{m_g}^{H}\mathbf{w}_g|^2}{\chi_{s,m_g}(\mathbf{w},\mathbf{U})}\Bigr) = - \ln\Bigr(1 - \frac{|\mathbf{h}_{m_g}^{H}\mathbf{w}_g|^2}{\chi_{s,m_g}(\mathbf{w},\mathbf{U})+|\mathbf{h}_{m_g}^{H}\mathbf{w}_g|^2}\Bigr). 
\end{equation}
From the fact that $0\leq\frac{|\mathbf{h}_{m_g}^{H}\mathbf{w}_g|^2}{\chi_{s,m_g}(\mathbf{w},\mathbf{U})+|\mathbf{h}_{m_g}^{H}\mathbf{w}_g|^2}\triangleq \Theta(\mathbf{w},\mathbf{U}) < 1$, the function $- \ln\bigr(1 - \Theta(\mathbf{w},\mathbf{U})\bigr)$ is jointly convex w.r.t. the involved variables \cite{Stephen}, which is useful for developing an approximate solution for \eqref{eq:srlog:1}. In particular, at feasible point $\bigl(\mathbf{w}^{(n)},\mathbf{U}^{(n)}\bigr)$, applying \eqref{inq1} yields
\begin{IEEEeqnarray}{rCl}
&&- \ln\Bigr(1 - \frac{|\mathbf{h}_{m_g}^{H}\mathbf{w}_g|^2}{\chi_{s,m_g}(\mathbf{w},\mathbf{U})+|\mathbf{h}_{m_g}^{H}\mathbf{w}_g|^2}\Bigr) \nonumber\\
&&\qquad\geq - \ln\Bigr(1 - \frac{|\mathbf{h}_{m_g}^{H}\mathbf{w}_g^{(n)}|^2}{\chi_{s,m_g}(\mathbf{w}^{(n)},\mathbf{U}^{(n)})+|\mathbf{h}_{m_g}^{H}\mathbf{w}_g^{(n)}|^2}\Bigr)\nonumber\\
&&\qquad\quad-\ \Gamma_{s,m_g}\bigr(\mathbf{w}^{(n)},\mathbf{U}^{(n)}\bigr)+2\frac{\Re\bigl\{(\mathbf{w}_g^{(n)})^H\mathbf{h}_{m_g}
\mathbf{h}_{m_g}^{H}\mathbf{w}_g\bigr\}}{\chi_{s,m_g}(\mathbf{w}^{(n)},\mathbf{U}^{(n)})}\nonumber\\
&&\qquad\quad-\  \frac{\Gamma_{s,m_g}\bigr(\mathbf{w}^{(n)},\mathbf{U}^{(n)}\bigr)\Bigl(\chi_{s,m_g}(\mathbf{w},\mathbf{U})+|\mathbf{h}_{m_g}^{H}\mathbf{w}_g|^2\Bigr)}{
\chi_{s,m_g}(\mathbf{w}^{(n)},\mathbf{U}^{(n)})+|\mathbf{h}_{m_g}^{H}\mathbf{w}_g^{(n)}|^2} \nonumber\\
&&\qquad:= \mathcal{F}_{m_g}^{(n)}(\mathbf{w},\mathbf{U}).
\label{eq:srlog:2}
\end{IEEEeqnarray}
Note that $\mathcal{F}_{m_g}^{(n)}(\mathbf{w},\mathbf{U})$ is concave and is global lower bound of $- \ln\bigr(1 - \Theta(\mathbf{w},\mathbf{U})\bigr)$. 
It implies that we can iteratively replace $- \ln\bigr(1 - \Theta(\mathbf{w},\mathbf{U})\bigr)$ by $\mathcal{F}_{m_g}^{(n)}(\mathbf{w},\mathbf{U})$ to achieve a convex approximation of \eqref{eq:rew:2:b} \cite{Marks:78}. Hence, by substituting \eqref{f1sinr1}, \eqref{eq:srlog:1}, and \eqref{eq:srlog:2} into \eqref{eq:rew:2:b}, we provide \eqref{eq:rew:2:b:equi}. 
To handling the constraint \eqref{eq:rew:2:c}, we equivalently rewrite $\Gamma_{e,k_g}(\mathbf{w},\mathbf{U})$ as
\begin{equation}\label{eq:sesinr:1}
\Gamma_{e,k_g}(\mathbf{w},\mathbf{U})=\frac{|\mathbf{g}_{k_g}^{H}\mathbf{w}_g|^2}{\chi_{e,k_g}(\mathbf{w},\mathbf{U})}
\end{equation}
where
\[
\chi_{e,k_g}(\mathbf{w},\mathbf{U})=\sum_{i=1,i\neq g}^G|\mathbf{g}_{k_g}^{H}\mathbf{w}_i|^2+\|\mathbf{g}_{k_g}^{H}\mathbf{U}\|^2+P_p|f_{k_g}|^2+\sigma_{k_g}^2.
\]
 The constraint \eqref{eq:rew:2:c} requires a tight upper bound of $\log_2\bigl(1+\Gamma_{e,k_g}(\mathbf{w},\mathbf{U})\bigl)$. Applying \eqref{inq3} yields
\begin{IEEEeqnarray}{rCl}\label{eq:sesinr:3:temp1}
&&\ln\bigl(1+\Gamma_{e,k_g}(\mathbf{w},\mathbf{U})\bigl) \leq \log_2\bigl(1+\Gamma_{e,k_g}(\mathbf{w}^{(n)},\mathbf{U}^{(n)})\bigl)\nonumber\\
&&\qquad\qquad\quad +\; \bigl(1+\Gamma_{e,k_g}(\mathbf{w}^{(n)},\mathbf{U}^{(n)})\bigl)^{-1}\nonumber\\
&&\qquad\qquad\quad \times\Bigl(\frac{|\mathbf{g}_{k_g}^{H}\mathbf{w}_g|^2}{\chi_{e,k_g}(\mathbf{w},\mathbf{U})} - \Gamma_{e,k_g}(\mathbf{w}^{(n)},\mathbf{U}^{(n)})\Bigr).
\end{IEEEeqnarray}
Although the right-hand side of \eqref{eq:sesinr:3:temp1} is still nonconvex, it can be further convexified  by
\begin{IEEEeqnarray}{rCl}\label{eq:sesinr:4}
&&\mathcal{F}_{k_g}^{(n)}(\mathbf{w},\mathbf{U}) :=  \log_2\bigl(1+\Gamma_{e,k_g}(\mathbf{w}^{(n)},\mathbf{U}^{(n)})\bigl)\nonumber\\
&&\qquad\qquad\quad +\; \bigl(1+\Gamma_{e,k_g}(\mathbf{w}^{(n)},\mathbf{U}^{(n)})\bigl)^{-1}\nonumber\\
&&\qquad\qquad\quad \times\Bigl(\frac{|\mathbf{g}_{k_g}^{H}\mathbf{w}_g|^2}{\Phi_{k_g}^{(n)}(\mathbf{w},\mathbf{U})} - \Gamma_{e,k_g}(\mathbf{w}^{(n)},\mathbf{U}^{(n)})\Bigr)
\end{IEEEeqnarray}
where $\Phi_{k_g}^{(n)}(\mathbf{w},\mathbf{U})$ is the first-order approximation of $\chi_{e,k_g}(\mathbf{w},\mathbf{U})$  around the point $(\mathbf{w}^{(n)},\mathbf{U}^{(n)})$ by using \eqref{inq2}, which is given by
\begin{IEEEeqnarray}{rCl}
\Phi_{k_g}^{(n)}(\mathbf{w},\mathbf{U}) \triangleq&& \sum_{i=1,i\neq g}^G 2\Re\bigl\{(\mathbf{w}_i^{(n)})^H\mathbf{g}_{k_g}\mathbf{g}_{k_g}^{H}\mathbf{w}_i\bigr\}  \nonumber\  \\
&&-\; \sum_{i=1,i\neq g}^G|\mathbf{g}_{k_g}^{H}\mathbf{w}_i^{(n)}|^2 + 2\Re\bigl\{\mathbf{g}_{k_g}^{H}\mathbf{U}^{(n)}\mathbf{U}^H\mathbf{g}_{k_g} \bigr\}    \nonumber\\
 &&-\;\|\mathbf{g}_{k_g}^{H}\mathbf{U}^{(n)}\|^2 + P_p|f_{k_g}|^2+\sigma_{k_g}^2.       \nonumber
\end{IEEEeqnarray}
The constraint \eqref{eq:rew:2:c} is then approximated by the following convex constraint
\begin{IEEEeqnarray}{rCl}\label{eq:sesinr:3:temp}
\mathcal{F}_{k_g}^{(n)}(\mathbf{w},\mathbf{U}) &\leq&  t_g\ln2.
\end{IEEEeqnarray}

In a similar manner, at feasible point $\bigl(\mathbf{w}^{(n)},\mathbf{U}^{(n)}\bigr)$, the nonconvex constraints \eqref{eq:rew:2:d} and  \eqref{eq:rew:2:e} are approximated by the following convex constraints
\begin{IEEEeqnarray}{rCl}\label{eq:prsinr:11}
\mathcal{P}_{l}^{(n)}(\mathbf{w},\mathbf{U}) &\geq&  (z + \bar{R}_{p,l})\ln2\label{eq:prsinr:1a1},\\
\mathcal{P}_{k_p}^{(n)}(\mathbf{w},\mathbf{U}) &\leq&  z\ln2\label{eq:prsinr:1b1}
\end{IEEEeqnarray}
where $\mathcal{P}_{l}^{(n)}(\mathbf{w},\mathbf{U})$ and $\mathcal{P}_{k_p}^{(n)}(\mathbf{w},\mathbf{U})$ are respectively given by
\begin{IEEEeqnarray}{rCl}
&&\mathcal{P}_{l}^{(n)}(\mathbf{w},\mathbf{U}) :=
 \ln\Bigr(1 + \Gamma_{p,l}(\mathbf{w}^{(n)},\mathbf{U}^{(n)})\Bigr) 
+  \Gamma_{p,l}\bigr(\mathbf{w}^{(n)},\mathbf{U}^{(n)}\bigr)\nonumber\\
&&\qquad\quad -\  \Gamma_{p,l}(\mathbf{w}^{(n)},\mathbf{U}^{(n)})\frac{\left(\chi_{p,l}(\mathbf{w},\mathbf{U})+P_p|h_l|^2\right)}{\chi_{p,l}(\mathbf{w}^{(n)},\mathbf{U}^{(n)})+P_p|h_l|^2}\label{eq:prsinr:2},\\
&&\mathcal{P}_{k_p}^{(n)}(\mathbf{w},\mathbf{U}) :=
 \ln\Bigr(1 + \Gamma_{e,k_p}(\mathbf{w}^{(n)},\mathbf{U}^{(n)})\Bigr) \nonumber\\
&&\qquad\qquad\ +\;  \bigl(1+\Gamma_{e,k_p}(\mathbf{w}^{(n)},\mathbf{U}^{(n)})\bigr)^{-1}\nonumber\\
&&\qquad\qquad\  \times\Bigl(\frac{P_p|g_{k_p}|^2}{\Phi_{k_p}^{(n)}(\mathbf{w},\mathbf{U})}-\Gamma_{e,k_p}(\mathbf{w}^{(n)},\mathbf{U}^{(n)})\Bigr),\label{eq:prsinr:3}
\end{IEEEeqnarray}
with
\begin{IEEEeqnarray}{rCl}
&&\Phi_{k_p}^{(n)}(\mathbf{w},\mathbf{U}) = \sum_{g=1}^G 2\Re\{(\mathbf{w}_g^{(n)})^H\mathbf{f}_{k_p}\mathbf{f}_{k_p}^{H}\mathbf{w}_g\} - \sum_{g=1}^G|\mathbf{f}_{k_p}^{H}\mathbf{w}_g^{(n)}|^2 \nonumber\\
    &&\qquad\qquad\qquad +\; 2\Re\{\mathbf{f}_{k_p}^{H}\mathbf{U}^{(n)}\mathbf{U}^H\mathbf{f}_{k_p}\} - \|\mathbf{f}_{k_p}^{H}\mathbf{U}^{(n)}\|^2+ \sigma_{k_p}^2, \nonumber\\
&&\chi_{p,l}(\mathbf{w},\mathbf{U})=\sum\nolimits_{g=1}^G|\mathbf{f}_l^{H}\mathbf{w}_g|^2+ \|\mathbf{f}_l^{H}\mathbf{U}\|^2+\sigma_l^2,\nonumber\\
&&\chi_{e,k_p}(\mathbf{w},\mathbf{U})=\sum\nolimits_{g=1}^G|\mathbf{f}_{k_p}^{H}\mathbf{w}_g|^2+ \|\mathbf{f}_{k_p}^{H}\mathbf{U}\|^2+\sigma_{k_p}^2.\nonumber
\end{IEEEeqnarray}

\section*{Appendix~B\\Proof of Proposition~\ref{prop1}} \label{Appendix:A}
Let  $\varphi(\mathbf{w},\mathbf{U})$ and $\varphi^{(n)}(\mathbf{w},\mathbf{U})$ denote the objective of \eqref{eq:rew:2} and \eqref{eq:convexapp:1}, respectively.
We have
\begin{equation}
\varphi(\mathbf{w},\mathbf{U})\geq
\varphi^{(n)}(\mathbf{w},\mathbf{U}),\ \ (\text{thanks to \eqref{eq:srlog:2}})
\end{equation}
and
\begin{equation}
\varphi(\mathbf{w}^{(n)},\mathbf{U}^{(n)})=
\varphi^{(n)}(\mathbf{w}^{(n)},\mathbf{U}^{(n)}),\ \ (\text{thanks to \eqref{eq:srlog:3}}).
\end{equation}
Let $\bigr(\mathbf{w}^{(n+1)},\mathbf{U}^{(n+1)}\bigr)$ and $\bigr(\mathbf{w}^{(n)},\mathbf{U}^{(n)}\bigr)$ be the optimal solution and feasible point of \eqref{eq:convexapp:1}, respectively. It follows that
\begin{IEEEeqnarray}{rCl}
\varphi\bigr(\mathbf{w}^{(n+1)},\mathbf{U}^{(n+1)}\bigr)&\geq& \varphi^{(n)}\bigr(\mathbf{w}^{(n+1)},\mathbf{U}^{(n+1)}\bigr)\nonumber\\
&\geq&\varphi^{(n)}\bigr(\mathbf{w}^{(n)},\mathbf{U}^{(n)}\bigr)\nonumber\\
&=&\varphi\bigr(\mathbf{w}^{(n)},\mathbf{U}^{(n)}\bigr).\label{eq:appA}
\end{IEEEeqnarray}
 It shows that $\bigr(\mathbf{w}^{(n+1)},\mathbf{U}^{(n+1)}\bigr)$ is a better point to \eqref{eq:convexapp:1}  than $\bigr(\mathbf{w}^{(n)},\mathbf{U}^{(n)}\bigr)$ in the scene of improving the objective value. Furthermore, the sequence $\{\varphi^{(n)}\}$ is bounded above due to the power constraint in \eqref{eq:8c}. Let $\bigr(\bar{\mathbf{w}},\bar{\mathbf{U}}\bigr)$  be a saddle point of \eqref{eq:convexapp:1}, by Cauchy's theorem, there is a
convergent subsequence $\bigr\{\bigr(\mathbf{w}^{(n_{\kappa})},\mathbf{U}^{(n_{\kappa})}\bigr)\bigr\}$ satisfying
\begin{equation}\label{eq:appA:1}
\lim_{\kappa\rightarrow +\infty}\left[\varphi\bigr(\mathbf{w}^{(n_{\kappa})},\mathbf{U}^{(n_{\kappa})}\bigr)-
\varphi\bigr(\bar{\mathbf{w}},\bar{\mathbf{U}}\bigr)\right]=0.
\end{equation}
For every $n$ there is $\kappa$ such that $n_{\kappa}\leq n\leq n_{\kappa+1}$. From \eqref{eq:appA} and \eqref{eq:appA:1}, it is true that
\begin{IEEEeqnarray}{rCl}
0 &=&  \lim_{\kappa\rightarrow +\infty}\left[\varphi\bigr(\mathbf{w}^{(n_{\kappa})},\mathbf{U}^{(n_{\kappa})}\bigr)-
\varphi\bigr(\bar{\mathbf{w}},\bar{\mathbf{U}}\bigr)\right]\nonumber\\
&\leq& \lim_{n\rightarrow+\infty}\left[\varphi\bigr(\mathbf{w}^{(n)},\mathbf{U}^{(n)}\bigr)-
\varphi\bigr(\bar{\mathbf{w}},\bar{\mathbf{U}}\bigr)\right] \nonumber\\
& \leq& \lim_{\kappa\rightarrow+\infty} \left[\varphi\bigr(\mathbf{w}^{(n_{\kappa + 1})},\mathbf{U}^{(n_{\kappa +1})}\bigr)-
\varphi\bigr(\bar{\mathbf{w}},\bar{\mathbf{U}}\bigr)\right]\nonumber\\
& =& 0
\end{IEEEeqnarray}
which leads to $\underset{n\rightarrow+\infty}{\lim}\varphi\bigr(\mathbf{w}^{(n)},\mathbf{U}^{(n)}\bigr)=
\varphi\bigr(\bar{\mathbf{w}},\bar{\mathbf{U}}\bigr)$. In other words, Algorithm \ref{algo:proposed:DUAL} will stop when the following termination condition is met, i.e., 
\begin{equation}
\left|\left(\varphi\bigr(\mathbf{w}^{(n)},\mathbf{U}^{(n)}\bigr)-
\varphi\bigr(\bar{\mathbf{w}},\bar{\mathbf{U}}\bigr)\right)/\varphi\bigr(\bar{\mathbf{w}},\bar{\mathbf{U}}\bigr)\right|  \leq \epsilon
\end{equation}
where $\epsilon$ is a given accuracy. Following the same arguments as those in \cite[Theorem 1]{Marks:78}, we can prove that  each accumulation point $\bigr(\bar{\mathbf{w}},\bar{\mathbf{U}}\bigr)$ of the sequence $\bigr\{\bigr(\mathbf{w}^{(n)},\mathbf{U}^{(n)}\bigr)\bigr\}$ is a KKT-point  of \eqref{eq:problem_1}. Proposition \ref{prop1} is thus proved.

\section*{Appendix~C\\Proof of Lemma~\ref{lemma:PR}} \label{Appendix:B}
Since the channels are modeled as i.i.d. Rayleigh random variables, the constraint in \eqref{eq:imcsi:2:g} can be rewritten for each $k_p$ link as
\begin{IEEEeqnarray}{rCl}
&\Pr\Bigl(\frac{P_p|g_{k_p}|^2}{\sum_{g=1}^G\tr(\mathbf{F}_{k_p}\widetilde{\mathbf{W}}_g)+\tr(\mathbf{F}_{k_p}\widetilde{\mathbf{U}})+\sigma_{k_p}^2}\leq \beta\Bigr)\geq\tilde{\epsilon}\qquad\\
&\Leftrightarrow\Pr\Bigl(\frac{P_p}{\beta}|g_{k_p}|^2\leq\tr\bigl(\mathbf{F}_{k_p}\bigl(\sum_{g=1}^G\widetilde{\mathbf{W}}_g+\widetilde{\mathbf{U}}\bigr)\bigr)+\sigma_{k_p}^2\Bigr)\geq\tilde{\epsilon}\qquad
\label{eq:proba:pri:1}\end{IEEEeqnarray}
where $\mathbf{F}_{k_p} \triangleq \mathbf{f}_{k_p}\mathbf{f}_{k_p}^H$ and $\widetilde{\mathbf{W}}_g \triangleq \mathbf{w}_g\mathbf{w}_g^H$. It is very difficult to   calculate the distribution of $\tr\Bigl(\mathbf{F}_{k_p}\bigl(\sum_{g=1}^G\widetilde{\mathbf{W}}_g+\widetilde{\mathbf{U}}\bigr)\Bigr)$ directly. Instead of this,  we consider its lower bound. For notational simplicity, let us define $\mathbf{A}=\sum_{g=1}^G\widetilde{\mathbf{W}}_g+\widetilde{\mathbf{U}}$. In \cite{Lasserre},
\begin{equation}
\sum_{i=1}^N\lambda_i(\mathbf{F}_{k_p})\lambda_{N-i+1}(\mathbf{A})\leq\tr(\mathbf{F}_{k_p}\mathbf{A})
\label{eq:inequalityoftrace}\end{equation}
is shown  for $N\times N$ Hermitian matrices $\mathbf{F}_{k_p}$ and $\mathbf{A}$,
where $\lambda_i(\mathbf{X})$ denotes the $i$-th eigenvalue of matrix $\mathbf{X}\in\mathbb{H}^{N\times N}$, and its magnitude is sorted as $\lambda_{\max}(\mathbf{X})=\lambda_1(\mathbf{X})\geq\lambda_2(\mathbf{X})\geq\cdots\geq\lambda_N(\mathbf{X})=\lambda_{\min}(\mathbf{X})$. Since $\mathbf{F}_{k_p}$ is a rank-one positive semidefinite matrix, \eqref{eq:inequalityoftrace} can be written as
 \begin{IEEEeqnarray}{rCl}
\tr(\mathbf{F}_{k_p}\mathbf{A})&\geq&\lambda_1(\mathbf{F}_{k_p})\lambda_{N}(\mathbf{A})\nonumber\\
&=&\lambda_{\max}(\mathbf{F}_{k_p})\lambda_{\min}(\mathbf{A})\nonumber\\
&=&\tr(\mathbf{F}_{k_p})\lambda_{\min}(\mathbf{A}).
\label{eq:inequalityoftrace_revise}\end{IEEEeqnarray}
Substituting \eqref{eq:inequalityoftrace_revise} into \eqref{eq:proba:pri:1}, we get
\begin{IEEEeqnarray}{rCl}
&&\Pr\Bigl(\frac{P_p}{\beta}|g_{k_p}|^2\leq\tr\bigl(\mathbf{F}_{k_p}\bigl(\sum_{g=1}^G\widetilde{\mathbf{W}}_g+\widetilde{\mathbf{U}}\bigr)\bigr)+\sigma_{k_p}^2\Bigr)\nonumber\\
&&\geq\Pr\Bigl(\frac{P_p}{\beta}|g_{k_p}|^2\leq\tr(\mathbf{F}_{k_p})\lambda_{\min}(\mathbf{A})+\sigma_{k_p}^2\Bigr)\geq                            \tilde{\epsilon}.
\label{eq:proba:pri:2}\end{IEEEeqnarray}
Let $x=\tr(\mathbf{F}_{k_p})=\tr(|\mathbf{f}_{k_g}|^2)$. Then, $x$ follows a chi-squared distribution since $|\mathbf{f}_{k_g}|^2$ is a sum of  squares of $N$ independent Gaussian random variables. Correspondingly, the  probability density function (PDF) of $x$ is given as
$f_{X}(x)=\frac{e^{-x}x^{N-1}}{\Gamma(N)}.$
Let $y=\frac{P_p}{\beta}|g_{k_p}|^2$, and it then follows an exponential distribution with the PDF as  $f_{Y}(y)=\frac{\beta}{P_p}e^{-\frac{\beta}{P_p}y}.$
Therefore, the probability in \eqref{eq:proba:pri:2}  is obtained as 
\begin{IEEEeqnarray}{rCl}\label{eq:proba:pri:3}
&&\Pr\Bigl(y\leq x\lambda_{\min}(\mathbf{A})+\sigma_{k_p}^2\Bigr)\geq\tilde{\epsilon} \nonumber\\
&&\Leftrightarrow \int_{0}^{\infty}\int_{0}^{x\lambda_{\min}(\mathbf{A})+\sigma_{k_p}^2}f_X(x)f_Y(y)dydx\geq\tilde{\epsilon}\nonumber\\
&&\Leftrightarrow \int_{0}^{\infty}\Bigl(1-\exp\Bigl(-\frac{\beta}{P_p}(x\lambda_{\min}(\mathbf{A})+\sigma_{k_p}^2)\Bigr)\Bigr)f_X(x)dx\geq\tilde{\epsilon}\nonumber\\
&&\stackrel{(a)}{\Leftrightarrow}1-\exp\Bigl(-\frac{\beta}{P_p}\sigma_{k_p}^2\Bigr)\Bigl[\frac{\beta}{P_p}\lambda_{\min}(\mathbf{A})+1\Bigr]^{-N}\geq\tilde{\epsilon}
\end{IEEEeqnarray}
where $(a)$ is obtained using~\cite[Eq.~(3.351.3)]{Gradsh}. Next, the constraint in \eqref{eq:imcsi:2:g} for $K_p$ links is given as
\begin{equation}\begin{aligned}
&\eqref{eq:imcsi:2:g}\Leftrightarrow \\
&\prod_{k_p=1}^{K_p}\Pr\Bigl(\frac{P_p|g_{k_p}|^2}{\sum_{g=1}^G\tr(\mathbf{F}_{k_p}\widetilde{\mathbf{W}}_g)+\tr(\mathbf{F}_{k_p}\widetilde{\mathbf{U}})+\sigma_{k_p}^2}\leq \beta\Bigr)\geq\tilde{\epsilon}\\
&\stackrel{(b)}{\Leftrightarrow}1-\exp\Bigl(-\frac{\beta}{P_p}\sigma_{k_p}^2\Bigr)\Bigl[\frac{\beta}{P_p}\lambda_{\min}(\mathbf{A})+1\Bigr]^{-N}\geq\tilde{\epsilon}^{1/K_p}\\
&\Leftrightarrow\lambda_{\min}(\mathbf{A})\geq\Bigl[\exp\bigl(-\frac{\beta}{NP_p}\sigma_{k_p}^2\bigr)/(1-\tilde{\epsilon}^{1/K_p})^{1/N}-1\Bigl]\frac{P_p}{\beta}
\end{aligned}\label{eq:proba:pri:4}\end{equation}
where $(b)$ is obtained by combining \eqref{eq:proba:pri:3} since the channels of $K_p$ passive Eves are independent and modeled as i.i.d. random variables.

\section*{Appendix~D\\Proof of Lemma~\ref{lemma:SR}} \label{Appendix:C}

 The constraint in \eqref{eq:imcsi:2:c} can be rewritten for each $k_g$ link as
\begin{IEEEeqnarray}{rCl}\label{eq:proba:second:1}
&&\Pr\Bigl(P_p\phi_g|f_{k_g}|^2\geq \nonumber\\
&&\tr\bigl(\mathbf{G}_{k_g}\bigl(\widetilde{\mathbf{W}}_g-\phi_g\sum_{i=1,i\neq g}^G\widetilde{\mathbf{W}}_i-\phi_g\widetilde{\mathbf{U}}\bigl)\bigl)-\sigma_{k_g}^2\phi_g \Bigl)\geq\epsilon_{g}  \quad\end{IEEEeqnarray}
where $\mathbf{G}_{k_g}\triangleq \mathbf{g}_{k_g}\mathbf{g}_{k_g}^H$ for all $k_g$. For any given $N\times N$ Hermitian matrix  $\mathbf{B}$,  it follows from \cite{Lasserre} that
\begin{IEEEeqnarray}{rCl}\label{eq:inequalityoftrace:2}
\tr(\mathbf{G}_{k_g}\mathbf{B})&\leq&\sum_{i=1}^N\lambda_i(\mathbf{G}_{k_g})\lambda_{i}(\mathbf{B})\nonumber\\
&=&\lambda_{\max}(\mathbf{G}_{k_g})\lambda_{\max}(\mathbf{B})\nonumber\\
&=&\tr(\mathbf{G}_{k_g})\lambda_{\max}(\mathbf{B}).
\end{IEEEeqnarray}
Substituting \eqref{eq:inequalityoftrace_revise} and \eqref{eq:inequalityoftrace:2} into \eqref{eq:proba:second:1}, we have
\begin{IEEEeqnarray}{rCl}\label{eq:proba:second:2}
&&\Pr\Bigl(P_p\phi_g|f_{k_g}|^2\geq
\tr\bigl(\mathbf{G}_{k_g}\bigl(\widetilde{\mathbf{W}}_g-\phi_g\sum_{i=1,i\neq g}^G\widetilde{\mathbf{W}}_i \nonumber\\
&&\qquad\qquad\qquad\qquad\qquad -\;\phi_g\widetilde{\mathbf{U}}\bigl)\bigl)-\sigma_{k_g}^2\phi_g \Bigl)\nonumber\\
&&\geq \Pr\Bigl(P_p\phi_g|f_{k_g}|^2\geq\tr(\mathbf{G}_{k_g})\Bigl[\|\mathbf{w}_g\|^2-\phi_g\sum_{i=1,i\neq g}^G\|\mathbf{w}_i\|^2 \nonumber\\
&&\qquad\qquad\qquad\qquad -\;  \phi_g\lambda_{\min}(\widetilde{\mathbf{U}})  \Bigr]-\sigma_{k_g}^2\phi_g \Bigl)\geq\epsilon_{g}.
\end{IEEEeqnarray}
Following similar steps to the proof of Lemma \ref{lemma:PR}, we can obtain 
\begin{IEEEeqnarray}{rCl}
&&\frac{\|\mathbf{w}_g\|^2}{\phi_g}-\sum_{i=1,i\neq g}^G\|\mathbf{w}_i\|^2-\lambda_{\min}(\widetilde{\mathbf{U}})  \nonumber\\
&&\qquad\qquad \leq \Bigl[\exp\Bigl(\frac{\sigma_{k_g}^2}{NP_p}\Bigr)\epsilon_{g}^{-1/NK_g}-1\Bigr]P_p\nonumber\\
&&\Leftrightarrow \frac{\|\mathbf{w}_g\|^2}{\phi_g} \leq \Bigl[\exp\Bigl(\frac{\sigma_{k_g}^2}{NP_p}\Bigr)\epsilon_{g}^{-1/NK_g}-1\Bigr]P_p \nonumber\\
&&\qquad\qquad +\; \sum_{i=1,i\neq g}^G\|\mathbf{w}_i\|^2 + \lambda_{\min}(\widetilde{\mathbf{U}})
\end{IEEEeqnarray}
which completes the proof.
\bibliographystyle{IEEEtran}
\bibliography{IEEEfull}

\end{document}